\documentclass[a4paper,french,english]{amsart}
\usepackage{mathptmx}

\usepackage[T1]{fontenc}
\usepackage[latin9]{inputenc}
\usepackage{units}
\usepackage{amsthm}
\usepackage{amstext}
\usepackage{amssymb}

\makeatletter

\pdfpageheight\paperheight
\pdfpagewidth\paperwidth

\providecommand{\tabularnewline}{\\}

\numberwithin{equation}{section}
\numberwithin{figure}{section}
\theoremstyle{plain}
\newtheorem{thm}{\protect\theoremname}
  \theoremstyle{definition}
  \newtheorem{example}[thm]{\protect\examplename}
  \theoremstyle{definition}
  \newtheorem{defn}[thm]{\protect\definitionname}
  \theoremstyle{remark}
  \newtheorem{rem}[thm]{\protect\remarkname}
  \theoremstyle{plain}
  \newtheorem{prop}[thm]{\protect\propositionname}

\usepackage{pgf,tikz} 
\usetikzlibrary{arrows} 

\usepackage{marvosym}
\usepackage{nicefrac}

\makeatother

\usepackage{babel}
\makeatletter
\addto\extrasfrench{%
   \providecommand{\fg}{\ifdim\lastskip>\z@\unskip\fi~\frqq}%
}

\makeatother
  \addto\captionsenglish{\renewcommand{\definitionname}{Definition}}
  \addto\captionsenglish{\renewcommand{\examplename}{Example}}
  \addto\captionsenglish{\renewcommand{\propositionname}{Proposition}}
  \addto\captionsenglish{\renewcommand{\remarkname}{Remark}}
  \addto\captionsenglish{\renewcommand{\theoremname}{Theorem}}
  \addto\captionsfrench{\renewcommand{\definitionname}{Définition}}
  \addto\captionsfrench{\renewcommand{\examplename}{Exemple}}
  \addto\captionsfrench{\renewcommand{\propositionname}{Proposition}}
  \addto\captionsfrench{\renewcommand{\remarkname}{Remarque}}
  \addto\captionsfrench{\renewcommand{\theoremname}{Théorème}}
  \providecommand{\definitionname}{Definition}
  \providecommand{\examplename}{Example}
  \providecommand{\propositionname}{Proposition}
  \providecommand{\remarkname}{Remark}
\providecommand{\theoremname}{Theorem}

\begin{document}

\title{\textsf{Gonosomal algebra}}

\maketitle
\begin{center}
\textsf{\textbf{Richard Varro}} \medskip{}

\par\end{center}

\begin{center}
{\footnotesize{}Institut de Math\'ematiques et de Mod\`elisation
de Montpellier, Universit\'e de Montpellier, }
\par\end{center}{\footnotesize \par}

\begin{center}
{\footnotesize{}35095 Montpellier Cedex 5, France.}
\par\end{center}{\footnotesize \par}

\begin{center}
{\footnotesize{}Département de Math\'ematiques et Informatique Appliqu\'ees,
Universit\'e de Montpellier III, }
\par\end{center}{\footnotesize \par}

\begin{center}
{\footnotesize{}34199 Montpellier Cedex 5, France}
\par\end{center}{\footnotesize \par}

\begin{center}
\textsf{\Large{}\bigskip{}
} 
\par\end{center}

\textbf{\footnotesize{}Abstract. }{\footnotesize{}We introduce the
gonosomal algebra. Gonosomal algebra extend the evolution algebra
of the bisexual population (EABP) defined by Ladra and Rozikov. We
show that gonosomal algebras can represent algebraically a wide variety
of sex determination systems observed in bisexual populations. We
illustrate this by about twenty genetic examples, most of these examples
cannot be represented by an EABP. We give seven algebraic constructions
of gonosomal algebras, each is illustrated by genetic examples. We
show that unlike the EABP gonosomal algebras are not dibaric. We approach
the existence of dibaric function and idempotent in gonosomal algebras.}\textbf{\footnotesize{}
}{\footnotesize{}\medskip{}
 }{\footnotesize \par}

{\footnotesize{}2000 Mathematical Subject Classification. Primary
: 17D92.}{\footnotesize \par}

{\footnotesize{}Key words. Baric algebra, dibaric algebra, commutative
duplication, non commutative duplication, bisexual population, gonosomal
gene, sex determining systems.}{\footnotesize \par}

\section{Introduction}

To construct an algebraic model of sex-linked inheritance (i.e. controlled
by sex), it is necessary to solve two problems: the asymmetry of the
transmission of genes and the wide variety of sex-determination systems.\medskip{}

In populations with separate sexes, bisexual and diploid, the sexual
differentiation is controlled in the majority of cases by a pair of
chromosomes: the sex chromosomes also called gonosomes (or heterochromosomes,
or heterosomes). The gonosomes often have different shapes (dimorphism)
which induces an asymmetry in the transmission of sex-linked traits.

In the animal kingdom, sex is determined by different systems, we
distinguish five systems: $XY$, $WZ$, $X0$, $Z0$ and $WXY$, which
can be added their multiple variants.

In the $XY$ sex-determination system observed in most mammals, homogametic
$XX$ individuals are females and heterogametic $XY$ are males. The
$XY$-system has multiple variants as $X_{1}X_{2}Y$ and $XY_{1}Y_{2}$.
In the  $X_{1}X_{2}Y$-system, females are $X_{1}X_{1}X_{2}X_{2}$
and males $X_{1}X_{2}Y$ while in the system $XY_{1}Y_{2}$ females
are $XX$ and males $XY_{1}Y_{2}$. Systems with more gonosomes as
$X_{1}X_{2}X_{3}X_{4}X_{5}Y$ or $X_{1}X_{2}X_{3}X_{4}X_{5}Y_{1}Y_{2}Y_{3}Y_{4}Y_{5}$
have also been observed. 

The $WZ$-system is mainly found in birds, females are heterogametic
$WZ$ while males are homogametic $ZZ$. This system also accepts
multiple variants such that $W_{1}W_{2}Z$, $WZ_{1}Z_{2}$, $W_{1}W_{2}Z_{1}Z_{2}$.

The $X0$-system is mainly observed in Hymenoptera (bees, wasps, ...).
In this system sex is controlled by a single $X$ chromosome, females
have two $X$ chromosomes, their genotype is $XX$, while males have
one, they are $X0$ (zero indicates the absence of the second $X$
chromosome). It has been observed in several species of spiders the
following systems $X_{1}X_{2}0$, $X_{1}X_{2}X_{3}0$, $X_{1}X_{2}X_{3}0$,
$X_{1}X_{2}X_{3}X_{4}0$.

As for $Z0$ sex determination system mainly observed among Lepidoptera,
females are $Z0$ and males $ZZ$. 

Finally, the $WXY$-system observed in several species of tropical
fish is more complex: an individual having a gonosome $Y$ is male
unless it is coupled with a $W$ chromosome, so in this system individuals
with genotypes $XY$ or $YY$ are males and those of genotypes $XX$,
$WX$ and $WY$ are female. 

However we shall see later in this work through numerous examples
that there are many other sex determination systems.

\medskip{}

A gene is said sex-linked or gonosomal if it is located on the sex
chromosomes. Because to the dimorphism of gonosomes there are two
kinds of gonosomal gene. Indeed, in the $XY$ and $WZ$ systems two
parts on gonosomes are observed: one part homologous (or pseudo-autosomal)
where genes are common to both gonosome types and a differential part
where a locus located on a gonosome has no counterpart on the other.
There are therefore two types of gonosomal genes: a gene is pseudo-autosomal
(or partially sex-linked) if the locus is located on the homologous
parts; it is gonosomal (or completely sex-linked) when the locus is
on the differential part of the sex chromosomes (cf. \cite{Heuch-75}).\medskip{}

In this work, after recalling the algebraic models of sex-linked inheritance,
we give a genetic example that can not be represented by these algebras,
this leads us to extend these algebras and define the gonosomal algebras.
Then we give seven algebraic constructions of gonosomal algebras and
illustrating them with examples we see that these algebras can represent
algebraically a wide variety of genetic phenomena related to sex as{\footnotesize{}:}\emph{\footnotesize{}
}\emph{i)} temperature-dependent sex determination; \emph{ii}) sequential
hermaphrodism; \emph{iii}) androgenesis; \emph{iv}) parthenogenesis;
\emph{v}) gynogenesis; \emph{vi}) bacterial conjugation \emph{vii})
cytoplasmic inheritance; \emph{viii}) sex-linked lethal genes; \emph{ix})
multiple sex-chromosome systems; \emph{x}) heredity in the WXY-system;
\emph{xi}) heredity in the WZ-system with male feminization; \emph{xii})
XY-system with fertile XY-females; \emph{xiii}) X-linked sex-ratio
distorter; \emph{xiv}) kleptogenesis;\emph{ xv}) genetic processes
(mutation, recombination, transposition) influenced by sex; \emph{xvi})
heredity in ciliates; \emph{xvii}) genomic imprinting; \emph{xviii})
$X$-inactivation; \emph{xix}) sex determination by gonosome elimination;
\emph{xx}) sexual reproduction in triploid; \emph{xxi}) polygenic
sex determination; \emph{xxii}) cytoplasmic heredity. We show that
gonosomal algebras are not in general dibaric and we give conditions
for the existence of a dibaric function and an idempotent in these
algebras.

\section{Preliminaries}

There are three algebraic models for a gonosomal gene inheritance.\medskip{}

The first method was proposed by Etherington \cite{Ether-41} in the
case of a diallelic sex-linked gene in the $XY$-system. It has been
used by Gonshor in the case of a diallelic sex-linked gene with mutation
\cite{Gonsh-60} and for a multiallelic sex-linked gene \cite{Gonsh-65}.
This model has been also described and studied by W\"orz-Busekros
\cite{WB-74}, \cite{WB-75} in the general case of a multiallelic
gene linked to sex in the XY-system. Starting from a basis $\left(a_{i}\right)_{1\leq i\leq n}\cup\left\{ Y\right\} $
of an algebra $A$ equipped with the multiplication:
\[
a_{i}a_{j}=\sum_{r=1}^{n}\gamma_{ijr}a_{r},\quad a_{i}Y=\sum_{r=1}^{n}\gamma_{ir}a_{r}+\frac{1}{2}Y,\quad YY=0,
\]
with $\sum_{r=1}^{n}\gamma_{ijr}=1$, $\sum_{r=1}^{n}\gamma_{ir}=\frac{1}{2}$
for all $1\leq i,j\leq n$; where the structure constant $\gamma_{ijr}$
(resp. $\gamma_{ir}$) represents the frequency of gametic type $e_{r}$
produced by a female (resp. male) of genotype $a_{i}a_{j}$ (resp.
$a_{i}Y$). Let $\mathcal{Z}$ be a vector space, $\left(a_{i}\otimes a_{j},a_{i}\otimes Y\right)_{1\leq i\leq j\leq n}$
a basis of $\mathcal{Z}$ where $a_{i}\otimes a_{j}$, $a_{i}\otimes Y$
represent respectively female and male genotypes, and we define on
$\mathcal{Z}$ the commutative structure algebra:\foreignlanguage{french}{
\begin{eqnarray}
\left(a_{i}\otimes a_{j}\right)\left(a_{k}\otimes Y\right) & = & \sum_{0\leq r\leq s\leq n}\left(\gamma_{ijr}\gamma_{ks}+\gamma_{ijs}\gamma_{kr}\right)a_{r}\otimes a_{s}+\frac{1}{2}\sum_{1\leq r\leq n}\gamma_{ijr}a_{r}\otimes Y,\nonumber \\
\left(a_{i}\otimes a_{j}\right)\left(a_{p}\otimes a_{q}\right) & = & 0,\label{eq:loiWB}\\
\left(a_{i}\otimes Y\right)\left(a_{i}\otimes Y\right) & = & 0.\nonumber 
\end{eqnarray}
}

The algebra $\mathcal{Z}$ is the zygotic algebra for sex-linked inheritance.

\medskip{}

The second definition is due to Gonshor \cite{Gonsh-60}, it does
not depend on a given base, it is built from the gametic state by
the sex duplication method. Starting from a baric algebra $\left(A,\omega\right)$,
we define on the space $A\otimes A\oplus A$ the structure algebra:
\begin{equation}
\left(x\otimes y\oplus z\right)\left(x'\otimes y'\oplus z'\right)=\frac{1}{2}\left(xy\otimes z'+x'y'\otimes z\right)\oplus\frac{1}{2}\left(\omega\left(z'\right)xy+\omega\left(z\right)x'y'\right).\label{eq:Gonshor}
\end{equation}

The resulting algebra is called the sex-linked duplicate.

In \cite{WB-75}, W\"orz-Busekros showed that the Etherington and
Gonshor definitions are equivalent.\medskip{}

Recently Ladra and Rozikov \cite{LR-10} introduced a more general
definition. Starting from the canonical basis $\left(e_{1},\text{\dots},e_{n+\nu}\right)$
of $\mathbb{R}^{n+\nu}$, $e_{i}^{\left(f\right)}=e_{i}$, $i=1,\ldots,n$
and $e_{i}^{\left(m\right)}=e_{n+i}$, $i=1,\ldots,\nu$, they provide
$\mathbb{R}^{n+\nu}$ with the structure algebra:
\begin{eqnarray*}
e_{i}^{\left(f\right)}e_{p}^{\left(m\right)}=e_{p}^{\left(m\right)}e_{i}^{\left(f\right)} & = & \frac{1}{2}\left(\sum_{k=1}^{n}P_{ip,k}^{\left(f\right)}e_{k}^{\left(f\right)}+\sum_{l=1}^{\nu}P_{ip,l}^{\left(m\right)}e_{l}^{\left(m\right)}\right),\\
e_{i}^{\left(f\right)}e_{j}^{\left(f\right)} & = & 0,\quad\qquad1\leq i,j\leq n,\\
e_{p}^{\left(m\right)}e_{q}^{\left(m\right)} & = & 0,\quad\qquad1\leq p,q\leq\nu,\\
\sum_{k=1}^{n}P_{ip,k}^{\left(f\right)}=\sum_{l=1}^{\nu}P_{ip,l}^{\left(m\right)} & = & 1,\quad\left(1\leq i\leq n,1\leq p\leq\nu\right).
\end{eqnarray*}

This algebra is named evolution algebra of the bisexual population
by their authors.

\section{Gonosomal algebra, definition and examples}

In this section we give a definition of a sex-linked algebra that
extends the one given in \cite{LR-10}. This extension of the definition
of \cite{LR-10} finds its source in the following example.

\subsection{Introductory example.}
\begin{example}
\label{exa:hemophilia} \emph{Heredity of hemophilia}.

Hemophilia is a genetic disorder linked to the $X$ chromosome, it
is due to mutations in two genes located at the end of the long arm
of gonosome $X$. This is a lethal recessive genetic disease that
is lethal in the homozygous state, it follows that if $X^{h}$ denotes
the $X$ chromosome carrying hemophilia, there are only two female
genotypes: $XX$ and $XX^{h}$ (genotype $X^{h}X^{h}$ is lethal)
and two male genotypes: $XY$ and $X^{h}Y$. The results of the four
kinds of crosses are:
\begin{alignat*}{2}
XX\times XY & \rightarrowtail\tfrac{1}{2}XX,\tfrac{1}{2}XY; & \hspace{1cm}XX\times X^{h}Y & \rightarrowtail\tfrac{1}{2}XX^{h},\tfrac{1}{2}XY;\\
XX^{h}\times XY & \rightarrowtail\tfrac{1}{4}XX,\tfrac{1}{4}XX^{h},\tfrac{1}{4}XY,\tfrac{1}{4}X^{h}Y; & XX^{h}\times X^{h}Y & \rightarrowtail\tfrac{1}{3}XX^{h},\tfrac{1}{3}XY,\tfrac{1}{3}X^{h}Y.
\end{alignat*}

Algebraically we represent this by the following commutative algebra
defined on the basis $\left(e_{1},e_{2},\widetilde{e}_{1},\widetilde{e}_{2}\right)$
of a $\mathbb{R}$-vector space:
\begin{alignat*}{2}
e_{1}\widetilde{e}_{1} & =\tfrac{1}{2}e_{1}+\tfrac{1}{2}\widetilde{e}_{1}, & \hspace{1cm}e_{1}\widetilde{e}_{2} & =\tfrac{1}{2}e_{2}+\tfrac{1}{2}\widetilde{e}_{1},\\
e_{2}\widetilde{e}_{1} & =\tfrac{1}{4}e_{1}+\tfrac{1}{4}e_{2}+\tfrac{1}{4}\widetilde{e}_{1}+\tfrac{1}{4}\widetilde{e}_{2}, & e_{2}\widetilde{e}_{2} & =\tfrac{1}{3}e_{2}+\tfrac{1}{3}\widetilde{e}_{1}+\tfrac{1}{3}\widetilde{e}_{2},\\
e_{i}e_{j} & =0 & \widetilde{e}_{i}\widetilde{e}_{j} & =0,\hspace{1cm}\left(i,j=1,2\right).
\end{alignat*}

It is clear that this algebra is not an evolution algebra of the bisexual
population as defined in \cite{LR-10}.
\end{example}

\subsection{Definition of a gonosomal algebra.}

\textcompwordmark{}\medskip{}

The above example leads us to put the following definition.
\begin{defn}
\label{def:Alggonosom}A $K$-algebra $A$ is gonosomale if there
is a basis $\mathcal{B}=\left\{ e_{i};1\leq i\leq n\right\} \cup\left\{ \widetilde{e}_{p};1\text{\ensuremath{\leq}}p\text{\ensuremath{\leq}}m\right\} $
of $A$ verifying for all $1\text{\ensuremath{\leq}}i,j\text{\ensuremath{\leq}}n$
and $1\text{\ensuremath{\leq}}p,q\leq m$ : 
\begin{eqnarray*}
e_{i}\widetilde{e}_{p}=\widetilde{e}_{p}e_{i} & = & \sum_{k=1}^{n}\gamma_{ipk}e_{k}+\sum_{r=1}^{m}\widetilde{\gamma}_{ipr}\widetilde{e}_{r},\\
e_{i}e_{j} & = & 0,\\
\widetilde{e}_{p}\widetilde{e}_{q} & = & 0\\
\sum_{k=1}^{n}\gamma_{ipk}+\sum_{r=1}^{m}\widetilde{\gamma}_{ipr} & = & 1.
\end{eqnarray*}
The basis $\mathcal{B}$ is called gonosomal basis of $A$.
\end{defn}
In this definition the vectors of $\left(e_{i}\right)_{1\text{\ensuremath{\le}}i\text{\ensuremath{\le}}n}$
(resp. $\left(\widetilde{e}_{p}\right)_{1\text{\ensuremath{\le}}p\text{\ensuremath{\le}}m}$)
are interpreted as genetic types observed in females (resp. in males),
the structure constant $\gamma_{ipk}$ (resp. $\widetilde{\gamma}_{ipr}$)
represents the female (resp. male) proportion of type $e_{k}$ (resp.
$\widetilde{e}_{r}$) in the progeny of a female type $e_{i}$ with
a male type $\widetilde{e}_{p}$.
\begin{rem}
If in the definition \ref{def:Alggonosom} we have $\sum_{k=1}^{n}\gamma_{ipk}=\sum_{r=1}^{m}\widetilde{\gamma}_{ipr}$
for all $1\text{\ensuremath{\leq}}i\text{\ensuremath{\leq}}n$ and
$1\text{\ensuremath{\leq}}p\text{\ensuremath{\leq}}m$, then noting
$P_{ip,k}^{\left(f\right)}=2\gamma_{ipk}$ and $P_{ip,r}^{\left(m\right)}=2\widetilde{\gamma}_{ipr}$
we have $e_{i}\widetilde{e}_{p}=\sum_{k=1}^{n}P_{ip,k}^{\left(f\right)}e_{k}+\sum_{r=1}^{m}P_{ip,r}^{\left(m\right)}\widetilde{e}_{r}$
with $\sum_{k=1}^{n}P_{ip,k}^{\left(f\right)}e_{k}=\sum_{r=1}^{m}P_{ip,r}^{\left(m\right)}=1$,
we find the definition of the evolution algebra of the bisexual population
given in \cite{LR-10}.
\end{rem}
\smallskip{}

\begin{rem}
A gonosomal algebra is commutative by definition. It is not in general
associative, so for the algebra $A$ of example \ref{exa:hemophilia}
we have: $e_{1}\left(\widetilde{e}_{1}e_{2}\right)-\left(e_{1}\widetilde{e}_{1}\right)e_{2}=\frac{1}{4}\left(\widetilde{e}_{1}+\widetilde{e}_{2}\right)\neq0$. 

The gonosomal algebras form a new class of non-associative algebra,
indeed the algebra $A$ given in example\ref{exa:hemophilia}, is
not Lie because $e_{1}\left(\widetilde{e}_{1}e_{2}\right)+\widetilde{e}_{1}\left(e_{2}e_{1}\right)+e_{2}\left(e_{1}\widetilde{e}_{1}\right)\neq0$
and taking in $A$, $x=\frac{2}{3}e_{1}+\frac{1}{3}\widetilde{e}_{1}$
and $y=e_{1}$ we obtain $x^{2}\left(yx\right)-\left(x^{2}y\right)x=$$\frac{1}{54}\left(e_{1}+\widetilde{e}_{1}\right)$
therefore $A$ is not either Jordan, nor power associative because
$x^{2}x^{2}-x^{4}=-\frac{1}{162}\left(e_{1}+\widetilde{e}_{1}\right)$,
nor alternative because we have $x^{2}y-x\left(xy\right)=$$\frac{1}{36}\left(e_{1}+\widetilde{e}_{1}\right)$. 
\end{rem}

\subsection{Examples of gonosomal algebras.}

\textcompwordmark{}

\medskip{}

We will show that the gonosomal algebras allows to represent algebraically
a wide variety of genetic phenomena linked to sex. In these examples,
unless otherwise specified, the term genetic type represents as well
alleles, genotypes or collection of genes whose loci are on gonosomes.
We start with some unusual examples.
\begin{example}
\emph{Temperature-dependent sex determination}. 

In reptiles (snakes, crocodiles, turtles, lizards) we find two types
of sex determination, either a genotypic determination controlled
according to the species by $XY$- or $ZW$-system, either a determination
depending on the incubation temperature of eggs (or TSD: temperature-dependent
sex determination). TSD is observed in all species of crocodilians
and most turtles. TSD is controlled by three temperature ranges, the
eggs subject to feminizing temperatures (TF) (resp. masculinising
(TM)) give rise to 100\% or a majority of females (resp. males) and
those subject to transition temperatures (TRT) provide 50\% females
and 50\% males.

Algebraically, we consider the space $A$ with basis $\left(e_{i},\widetilde{e}_{i}\right)_{1\leq i\leq n}$
where $e_{i}$ (resp. $ $$\widetilde{e}_{i}$) are female (resp.
male) genetic types present in a population and subject to TSD. We
note $\tau_{1}$, $\tau_{2}$ and $\tau_{3}$ the probability that
eggs are incubated at TF, TM and TRT respectively, thus we have $\tau_{1},\tau_{2},\tau_{3}\geq0$
and $\tau_{1}+\tau_{2}+\tau_{3}=1$. For each $r=1,2$ ($r=1$ for
TF, $r=2$ for TM) we note $\mu_{r}$ and $\widetilde{\mu}_{r}$ respectively
the proportions of females and males arising from eggs placed in the
environment $r$, thus we have $\mu_{r}+\widetilde{\mu}_{r}=1$ with
$\mu_{1}>\widetilde{\mu}_{1}\geq0$ and $0\leq\mu_{2}<\widetilde{\mu}_{2}$.
Finally, if for all $1\leq i,p\leq n$ we denote $\theta_{ipk}$ the
egg proportion of $e_{k}$ type in the laying of a female $e_{i}$
crossed with a male $\widetilde{e}_{p}$, thus $\sum_{k=1}^{n}\theta_{ipk}=1$.
Then the space $A$ equipped with the following algebra structure:
\begin{eqnarray*}
e_{i}e_{j} & = & \widetilde{e}_{p}\widetilde{e}_{q}=0,\\
e_{i}\widetilde{e}_{p} & = & \left(\mu_{1}\tau_{1}+\mu_{2}\tau_{2}+\frac{1}{2}\tau_{3}\right)\sum_{k=1}^{n}\theta_{ipk}e_{k}+\left(\widetilde{\mu}_{1}\tau_{1}+\widetilde{\mu}_{2}\tau_{2}+\frac{1}{2}\tau_{3}\right)\sum_{k=1}^{n}\theta_{ipk}\widetilde{e}_{k}
\end{eqnarray*}
is a gonosomal algebra and the product $ $$e_{i}\widetilde{e}_{p}$
gives the genetic distribution of progeny of a female $e_{i}$ with
a male $\widetilde{e}_{p}$.
\end{example}
\smallskip{}

\begin{example}
\emph{Sequential hermaphroditism.}

It is observed in many species of fish that the sex of an individual
changes during his life. We distinguish two cases of sequential hermaphroditism:
protogyny and protoandry. In protogyny (9\% of fish families) individuals
are first female then become male, the opposite occurs in the protoandry
(1\% of fish families) when a male change sex to female. This sex
change occurs as a result of the disappearance of the dominant male
or female.

This situation can be represented by a gonosomal algebra. For this
we consider in a protogynous hermaphrodite population, a gene whose
types are noted $e_{1},\ldots,e_{n}$ when they are observed in females
and $ $$\widetilde{e}_{1},\ldots,\widetilde{e}_{n}$ in males. The
progeny of a female type $e_{i}$ with a male $\widetilde{e}_{j}$
consists of individuals having a female phenotype. Therefore if $\lambda_{ijk}$
is the frequency of type $e_{k}$ in the progeny of the cross between
$e_{i}$ with $\widetilde{e}_{j}$ before sexual inversion, then the
distribution of types in offspring is $\sum_{k=1}^{n}\lambda_{ijk}e_{k}$.
And if $0<\theta_{k}<1$ denotes the rate of sexual invertion by generation
of a female $e_{k}$, the distribution of types becomes $e_{i}\widetilde{e}_{j}=\widetilde{e}_{j}e_{i}=\sum_{k=1}^{n}\left(1-\theta_{k}\right)\lambda_{ijk}e_{k}+\sum_{k=1}^{n}\theta_{k}\lambda_{ijk}\widetilde{e}_{k}$
. 
\end{example}
\smallskip{}

\begin{example}
\emph{Androgenesis.}

Androgenesis is a rare process observed in fish, molluscs and insects.
Androgenesis is the production of an offspring containing exclusively
the nuclear genome of the fathering male via the maternal eggs. During
fertilization spermatozoon penetrates the ovum which causes the expulsion
of the female genome. In some species the eggs give haploid individuals,
in others chromosomes are immediately duplicated after fertilization
to give diploid individuals and in other species, there is a polyspermy
that is to say, the ovum is simultaneously fertilized by several spermatozoa,
in this case we obtain diploid or polyploid individuals.

For example, in bivalves of genus \emph{Corbiculla,} androgenesis
is observed in four species which are simultaneous hermaphrodite,
i.e. male and female reproductive organs are mature at the same time.
In these species ova are fertilized by unreduced spermatozoa, i.e.
whose genetic composition is identical to that of the somatic cells
of the male transmitter. This mode of reproduction allows the coexistence
and interbreeding of individuals di-, tri- and tetraploid within the
same species of \emph{Corbiculla}.

Algebraically, let be $e_{1},\ldots,e_{n}$ and $\widetilde{e}_{1},\ldots,\widetilde{e}_{n}$
respectively the ova and spermatozoa genetic types carried by \emph{Corbiculla}.
Let $\theta_{i}$ the proportion of eggs among the gametes of an individual
of type $i$, we define on the space $A$ with basis $\left(e_{i},\widetilde{e}_{i}\right)_{1\leq i\leq n}$
the products $e_{i}\widetilde{e}_{j}=\widetilde{e}_{j}e_{i}=\theta_{j}e_{j}+\left(1-\theta_{j}\right)\widetilde{e}_{j}$,
other products are zero, then the algebra $A$ is gonosomal. 
\end{example}
\smallskip{}

\begin{example}
\emph{Sex determination by thelytokous parthenogenesis or by gynogenesis.}

In some species a female gamete develops an embryo without fertilization
by a male gamete, this form of asexual reproduction is called parthenogenesis.
There are several forms of parthenogenesis, one is the thelytoky.
In thelytokous parthenogenesis a diploid female gives birth only to
diploid females and therefore the population consists solely of females. 

Another phenomenon related to the thelytokous parthenogenesis is gynogenesis
(or pseudogamy or merospermy) which is observed in fish, amphibians
and insects. Gynogenesis requires the fertilization of the ovum by
a sperm cell of a close species what activates its development in
zygote, the male genome degenerates or is eliminated, thus it is not
expressed in the offspring which consists only of females.

To represent algebraically these two situations, we consider the gonosomal
$\mathbb{R}$-algebra defined on a basis $\left(e_{i}\right)_{1\leq i\leq n}\cup\left(\widetilde{e}\right)$
by $e_{i}\widetilde{e}=\sum_{i=1}^{n}\mu_{ik}e_{k}$ where $e_{1},\ldots,e_{n}$
symbolize female genotypes and $\mu_{ik}$ the mutation rate from
type $e_{i}$ to type $e_{k}$.
\end{example}
\smallskip{}

\begin{example}
\emph{Bacterial conjugation}.
\end{example}
Bacterial conjugation is a transfer of genetic material between bacterial
cells, a donor (called male) to a recipient (called female), that
is why bacterial conjugation is often regarded as the bacterial equivalent
of sexual reproduction. It is controlled by a conjugative factor carried
by a plasmid. Plasmids are double-stranded, circular DNA molecules,
present in the cytoplasm of bacteria that replicate autonomously.
Plasmids are not essential to a normal activity of bacteria but they
carry genes that provide a selective advantage to the holder (antibiotic
resistance, increased pathogenicity, bacteriocins synthesis inhibiting
the growth of other bacteria, acquisition of new metabolic properties).
In addition to these genes, some plasmids carry a conjugative factor
composed of several genes that control their transfer to another bacterial
cell. The most studied conjugative factor is the factor $F$, it contains
genes that encode the synthesis of pili allowing a bacterium $F^{+}$
to dock with a bacterium $F^{-}$, surface exclusion genes that prevent
two bacteria $F^{+}$ to moor and genes that allow the synthesis and
transfer of one copy of the plasmid into a bacterium $F^{-}$. At
the end of the conjugation, the factor $F$ persists in the donor
bacterium which stays $F^{+}$ and a copy of the plasmid carrying
this factor is acquired by the recipient bacterium which becomes $F^{+}$.

Algebraically, let $a_{1},\ldots,a_{n}$ be the different types of
bacterial chromosomes and $b_{1},\ldots,b_{m}$ the plasmid types
carrying a conjugative factor observed in a colony of bacteria. If
we put $e_{i}=a_{i}$ and $\widetilde{e}_{p,q}=a_{p}\otimes b_{q}$,
then the space $A$ with the basis $\left\{ e_{i}\right\} \cup\left\{ \widetilde{e}_{p,q}\right\} $
and the commutative product $e_{i}e_{j}=0$, $\widetilde{e}_{p,q}\widetilde{e}_{r,s}=0$,
$e_{i}\widetilde{e}_{p,q}=\frac{1}{2}\widetilde{e}_{p,q}+\frac{1}{2}\widetilde{e}_{i,q}$
is a gonosomal algebra which models the bacterial conjugation. If
$0\leq\tau\leq1$ denote conjugation rate, starting from two populations
of bacteria $x=\sum_{i}\alpha_{i}e_{i}+\sum_{p,q}\widetilde{\alpha}_{p,q}\widetilde{e}_{p,q}$
and $y=\sum_{i}\beta_{i}e_{i}+\sum_{p,q}\widetilde{\beta}_{p,q}\widetilde{e}_{p,q}$
where $\sum_{i}\alpha_{i}+\sum_{p,q}\widetilde{\alpha}_{p,q}=1$ and
$\sum_{i}\beta_{i}+\sum_{p,q}\widetilde{\beta}_{p,q}=1$, then the
population obtained after conjugation of $x$ and $y$ is given by
the product $x\star y=\tau xy+\frac{1-\tau}{2}\left(x+y\right)$.

\section{Gonosomal algebras constructions}

The definition of gonosomal algebra is very general, it allows to
represent a wide variety of sex-linked inheritance situations, but
the counterpart is that we can not give an intrinsic definition similar
to that of Gonshor for zygotic algebra of sex-linked inheritance.

In the following we give seven methods to construct gonosomal algebras.
These constructions are obtained:

-- by reduction of gonosomal algebra;

-- from the duplicate of a baric algebra;

-- from two baric algebras;

-- from two algebra laws;

-- from three linear forms and three linear maps;

-- from a baric algebra and two linear maps;

-- from a baric algebra and a gonosomal algebra.

\medskip{}

We can see each construction as a particular gonosomal algebra class
that can be studied for its own interest. Each construction is illustrated
by genetic examples.

\medskip{}

In the following constructions we often use the the notion of baric
algebra. A $K$-algebra $A$ is baric if it admits a non trivial algebra
morphism $\omega:A\rightarrow K$, called weight morphism of $A$,
we note $\left(A,\omega\right)$ to indicate that $A$ is weighted
by $\omega$ and for $x\in A$, the scalar $\omega\left(x\right)$
is called the weight of $x$. We also use the following result (cf.
\cite{WB-80}, lemma 1.10) : a finite-dimensional $K$-algebra $A$
is baric if and only if $A$ has a basis $\left(e_{1},\ldots,e_{n}\right)$
such that $e_{i}e_{j}=\sum_{k=1}^{n}\gamma_{ijk}e_{k}$ with $\gamma_{ijk}\in K$
and $\sum_{k=1}^{n}\gamma_{ijk}=1$ $\left(i,j=1,\ldots,n\right)$. 

\smallskip{}

\subsection{Construction by reduction of gonosomal algebra.}

\textcompwordmark{}

Starting from a gonosomal algebra we can construct others by reducing
the gonosomal basis. In the result below, for any integer $k\geq1$
we note $\left[\!\left[1,k\right]\!\right]=\left\{ 1,\ldots,k\right\} $. 
\begin{prop}
\label{pro:Reductionbase}Let $A$ be a gonosomal $K$-algebra, $\left(e_{i}\right)_{i\in\left[\left[1,n\right]\right]}\cup\left(\widetilde{e}_{p}\right)_{p\in\left[\!\left[1,m\right]\!\right]}$
a gonosomal basis of $A$ with $e_{i}\widetilde{e}_{p}=\sum_{k=1}^{n}\gamma_{ipk}e_{k}+\sum_{r=1}^{m}\widetilde{\gamma}_{ipr}\widetilde{e}_{r}$.
If there is $I\subsetneqq\left[\!\left[1,n\right]\!\right]$ and $J\subsetneqq\left[\!\left[1,m\right]\!\right]$
such that for all $i\in\left[\!\left[1,n\right]\!\right]\setminus I$
and $p\in\left[\!\left[1,m\right]\!\right]\setminus J$ we have $\sigma_{ip}=1-\sum_{k\in I}\gamma_{ipk}+\sum_{r\in J}\widetilde{\gamma}_{ipr}\neq0$
, then the subspace spanned by $\left(e_{i}\right)_{i\in\left[\!\left[1,n\right]\!\right]\setminus I}\cup\left(\widetilde{e}_{p}\right)_{p\in\left[\!\left[1,m\right]\!\right]\setminus J}$
with multiplication 
\[
e_{i}*\widetilde{e}_{p}=\sigma_{ip}^{-1}\left(\sum_{k\in\left[\!\left[1,n\right]\!\right]\setminus I}\gamma_{ipk}e_{k}+\sum_{r\in\left[\!\left[1,m\right]\!\right]\setminus J}\widetilde{\gamma}_{ipr}\widetilde{e}_{r}\right)
\]
 and $e_{i}*e_{j}=\widetilde{e}_{p}*\widetilde{e}_{q}=0$ for all
$\left(i,j\in\left[\!\left[1,n\right]\!\right]\setminus I;p,q\in\left[\!\left[1,m\right]\!\right]\setminus J\right)$,
is a gonosomal algebra.\end{prop}
\begin{proof}
This follows immediately from $\sum_{k\in\left[\!\left[1,n\right]\!\right]\setminus I}\gamma_{ipk}+\sum_{r\in\left[\!\left[1,m\right]\!\right]\setminus J}\widetilde{\gamma}_{ipr}=\sigma_{ip}$.
\end{proof}
This result is very useful to show that the transmision of a recessive
lethal character can be represented by a gonosomal algebra.
\begin{example}
\emph{Recessive lethal gonosomal allele in $XY$-system}.

A gene is lethal if one of its alleles causes the death of organisms
that carry them. We study in $XY$-system, a gonosomal gene having
two alleles $a$ and $b$. We are going to consider two cases. \smallskip{}

\textbf{Case 1}. Allele $b$ is recessive lethal to females.

In this case there is no genotype female $bb$. Therefore we have
only four crosses:

\hspace{2cm}%
\begin{tabular}{ll}
$aa\times aY\rightarrowtail\frac{1}{2}aa,\frac{1}{2}aY;$ & $aa\times bY\rightarrowtail\frac{1}{2}ab,\frac{1}{2}aY;$$\medskip$\tabularnewline
$ab\times aY\rightarrowtail\frac{1}{4}aa,\frac{1}{4}ab,\frac{1}{4}aY,\frac{1}{4}bY;$ & $ab\times bY\rightarrowtail\frac{1}{3}aa,\frac{1}{3}aY,\frac{1}{3}bY.$\tabularnewline
\end{tabular}

\smallskip{}

\textbf{Case} \textbf{2}. Allele $b$ is lethal to males.

In this case there is no genotype male $bY$. There are therefore
only three crosses:

\hspace{5mm}%
\begin{tabular}{ccc}
$aa\times aY\rightarrowtail\frac{1}{2}aa,\frac{1}{2}aY;\quad$ & $ab\times aY\rightarrowtail\frac{1}{3}aa,\frac{1}{3}ab,\frac{1}{3}aY;\quad$ & $bb\times aY\rightarrowtail ab.$\tabularnewline
\end{tabular}

\smallskip{}

Algebraically, we consider the basic algebra$A$ with basis $\left(e_{1},e_{2},e_{3},\widetilde{e}_{1},\widetilde{e}_{2}\right)$
defined for $i,j=1,2$, $i\neq j$ by
\begin{eqnarray*}
e_{i}\widetilde{e}_{i}=\widetilde{e}_{i}e_{i} & = & \tfrac{1}{2}e_{i}+\tfrac{1}{2}\widetilde{e}_{i},\\
e_{i}\widetilde{e}_{j}=\widetilde{e}_{j}e_{i} & = & \tfrac{1}{2}e_{3}+\tfrac{1}{2}\widetilde{e}_{i},\\
e_{3}\widetilde{e}_{i}=\widetilde{e}_{i}e_{3} & = & \tfrac{1}{4}e_{i}+\tfrac{1}{4}e_{3}+\tfrac{1}{4}\widetilde{e}_{1}+\tfrac{1}{4}\widetilde{e}_{2},
\end{eqnarray*}
other products are zero. It is clear that $A$ is a gonosomal algebra.
Using one hand, the correspondences $e_{1}\leftrightarrow aa$, $e_{2}\leftrightarrow bb$,
$e_{3}\leftrightarrow ab$, $\widetilde{e}_{1}\leftrightarrow aY$,
$\widetilde{e}_{2}\leftrightarrow bY$ and taking on the other hand
in the proposition \ref{pro:Reductionbase}, $I=\left\{ 2\right\} $,
$J=\textrm{Ø}$ we recognize the case 1, with $I=\textrm{Ø}$, $J=\left\{ 2\right\} $
we obtain the case 2.
\end{example}
\smallskip{}

\subsection{Construction from the duplicate of a baric algebra.}

$\qquad$\textcompwordmark{}

We recall that if $\left(A,\omega\right)$ is a commutative $K$-algebra,
the non commutative duplicate of $A$ is the space $A\otimes A$ and
the commutative duplicate of $A$ is the quotient space of $A\otimes A$
by the ideal spanned by $\left\{ x\otimes y-y\otimes x;x,y\in A\right\} $.
They are both noted $D\left(A\right)$ and equipped with the algebra
law: $\left(x\otimes y\right)\left(x'\otimes y'\right)=\left(xy\right)\otimes\left(x'y'\right)$.
The surjective morphism $\mu:D\left(A\right)\rightarrow A^{2}$, $x\otimes y\mapsto xy$
is called the Etherington morphism and the map $\omega_{D}=\omega\circ\mu$
is a weight of $D\left(A\right)$.
\begin{prop}
\label{pro:AGetDupliquee}Let $\left(A,\omega\right)$ be a finite-dimensional
baric commutative $K$-algebra and $A_{1}$, $A_{2}$ two vector subspaces
of $D\left(A\right)$ such that $A_{1},A_{2}\neq\left\{ 0\right\} $,
$A_{1},A_{2}\nsubseteq\ker\left(\omega_{D}\right)$, $A_{1}\cap A_{2}=\left\{ 0\right\} $
and $\mu\left(A_{1}\right)\otimes\mu\left(A_{2}\right)\subset A_{1}\oplus A_{2}$,
then the space $A_{1}\oplus A_{2}$ with multiplication 
\[
\left(x_{1}\oplus x_{2}\right)\left(y_{1}\oplus y_{2}\right)=\mu\left(x_{1}\right)\otimes\mu\left(y_{2}\right)+\mu\left(y_{1}\right)\otimes\mu\left(x_{2}\right)
\]
is a gonosomal algebra.\end{prop}
\begin{proof}
From $A_{1},A_{2}\neq\left\{ 0\right\} $ and $A_{1},A_{2}\nsubseteq\ker\left(\omega_{D}\right)$
we deduce there are $e_{1}\in A_{1}$ and $\widetilde{e}_{1}\in A_{2}$
such that $\omega_{D}\left(e_{1}\right)=\omega_{D}\left(\widetilde{e}_{1}\right)=1$.
Let us complete $\left\{ e_{1}\right\} $ into a basis $B=\left(e_{1},\ldots,e_{n}\right)$
of $A_{1}$, replacing in $B$ each element $e_{i}$ such that $\omega_{D}\left(e_{i}\right)\neq0$
by $\omega_{D}\left(e_{i}\right)^{-1}e_{i}$ and each term $e_{i}$
such that $\omega_{D}\left(e_{i}\right)=0$ by $e_{i}+e_{1}$, we
can suppose that we have $\omega_{D}\left(e_{i}\right)=1$ for all
$1\leq i\leq n$. Analogously we complete $\left\{ \widetilde{e}_{1}\right\} $
into a basis $\widetilde{B}=\left(\widetilde{e}_{1},\ldots,\widetilde{e}_{m}\right)$
of $A_{2}$ verifying $\omega_{D}\left(\widetilde{e}_{j}\right)=1$
for all $1\leq j\leq m$. By $A_{1}\cap A_{2}=\left\{ 0\right\} $
it follows that $B\cup\widetilde{B}$ is a basis of $A_{1}\oplus A_{2}$.
Let us show that $A_{1}\oplus A_{2}$ equipped with the product given
in the statement is gonosomal for this basis. From the multiplication
definition it occurs immediately that $e_{i}e_{j}=\widetilde{e}_{i}\widetilde{e}_{j}=0$.
Then for all $e_{i}\in B$ and $\widetilde{e}_{j}\in\widetilde{B}$
we have $e_{i}\widetilde{e}_{j}=\mu\left(e_{i}\right)\otimes\mu\left(\widetilde{e}_{j}\right)$,
but it follows from $\mu\left(e_{i}\right)\otimes\mu\left(\widetilde{e}_{j}\right)\in A_{1}\oplus A_{2}$
that $\mu\left(e_{i}\right)\otimes\mu\left(\widetilde{e}_{j}\right)=\sum_{k=1}^{n}\alpha_{ijk}e_{k}+\sum_{k=1}^{m}\beta_{ijk}\widetilde{e}_{k}$
and with $\omega_{D}\left(\mu\left(e_{i}\right)\otimes\mu\left(\widetilde{e}_{j}\right)\right)=\omega\left(\mu\left(e_{i}\right)\mu\left(\widetilde{e}_{j}\right)\right)=\omega_{D}\left(e_{i}\right)\omega_{D}\left(\widetilde{e}_{j}\right)=1$,
we obtain $\sum_{k=1}^{n}\alpha_{ijk}+\sum_{k=1}^{m}\beta_{ijk}=1$.\end{proof}
\begin{rem}
If in the proposition \ref{pro:AGetDupliquee}, $A$ is a $K$-algebra
with basis $\left(a_{i}\right)_{1\leq i\leq n}\cup\left(Y\right)$
weighted by $\omega\left(a_{i}\right)=\omega\left(Y\right)=1$ and
$A_{1}$, $A_{2}$ are subspaces of $D\left(A\right)$ with respective
basis $\left(a_{i}\otimes a_{j}\right)_{1\leq i\leq j\leq n}$ and
$\left(a_{i}\otimes Y\right)_{1\leq i\leq n}$. Then the Etherington
morphism $\mu$ gives the gametogenesis results for females $\mu\left(a_{i}\otimes a_{j}\right)=\sum_{k=1}^{n}\gamma_{ijk}a_{k}$
and for males $\mu\left(a_{i}\otimes Y\right)=\sum_{k=1}^{n}\gamma_{ik}a_{k}+\frac{1}{2}Y$,
with $\sum_{k=1}^{n}\gamma_{ijk}=\sum_{k=1}^{n}\gamma_{ik}+\frac{1}{2}=1$.
With this, the algebra law of $A_{1}\oplus A_{2}$ is 
\begin{eqnarray*}
\left(a_{i}\otimes a_{j}\right)\left(a_{k}\otimes Y\right) & = & \sum_{r,s=1}^{n}\gamma_{ijr}\gamma_{ks}a_{r}\otimes a_{s}+\frac{1}{2}\sum_{r=1}^{n}\gamma_{ijr}a_{r}\otimes Y\\
 & = & \sum_{0\leq r\leq s\leq n}\left(\gamma_{ijr}\gamma_{ks}+\gamma_{ijs}\gamma_{kr}\right)a_{r}\otimes a_{s}+\frac{1}{2}\sum_{1\leq r\leq n}\gamma_{ijr}a_{r}\otimes Y,
\end{eqnarray*}
we find the multiplication (\ref{eq:loiWB}).
\end{rem}
Sex determination in the great majority of species is controlled by
two genotypes (e.g. $XX$ or $XY$). However there are cases where
sex is encoded by more than two genotypes, we are going to give several
examples and see that these cases also obey the definition \ref{def:Alggonosom}.
\begin{example}
\emph{Heredity in the WXY-system.}

We consider a population whose sex is determined by the $WXY$-system,
in this system male genotypes are $XY$ or $YY$ and female $WX$,
$WY$ or $XX$. This situation is more complex than other systems
because here a gene can be completely gonosomal for a pair of sex
chromosomes and partially gonosomal (i.e. pseudo-autosomal) for another
pair. To take this into account in an algebraic model of this system
we introduce the following formalism.

Gonosomes $W$, $X$, $Y$ are denoted respectively by $\Gamma_{1}$,
$\Gamma_{2}$, $\Gamma_{3}$ . Considering a gonosomal gene, its alleles
are noted $a_{1},\ldots,a_{N}$ and $a_{0}$ is used to indicate that
the gene is not present on a gonosome.

For each $1\leq r\leq3$ we put $I_{r}=\left\{ 1,\ldots,N\right\} $
if gonosome $\Gamma_{r}$ carries the gene, otherwise we put $I_{r}=\left\{ 0\right\} $.
For each $i\in I_{r}$, the notation $a_{i}^{\left(r\right)}$ represents
a gamete containing gonosome $\Gamma_{r}$ and allele $a_{i}$, then
for all $i\in I_{r}$, $j\in I_{s}$ such that $r\leq s$, the genotype
of an individual of sex $\Gamma_{r}\Gamma_{s}$ with alleles $a_{i}$
on $\Gamma_{r}$ and $a_{j}$ on $\Gamma_{s}$ is denoted by $a_{i}^{\left(r\right)}a_{j}^{\left(s\right)}$.
Next we note $\alpha_{ijk}^{\left(r,s\right)}$ (resp. $\beta_{ijk}^{\left(r,s\right)}$)
the probability that an individual of genotype $a_{i}^{\left(r\right)}a_{j}^{\left(s\right)}$
produces a gamete carrying allele $a_{k}^{\left(r\right)}$ (resp.
$a_{k}^{\left(s\right)}$), thus we have $\sum_{k\in I_{r}}\alpha_{ijk}^{\left(r,s\right)}=\sum_{k\in I_{s}}\beta_{ijk}^{\left(r,s\right)}=1$
and we agree that $\alpha_{ij0}^{\left(r,s\right)}=0$ if $I_{r}\neq\left\{ 0\right\} $
and $\alpha_{0jk}^{\left(r,s\right)}=\alpha_{i0k}^{\left(r,s\right)}=\alpha_{00k}^{\left(r,s\right)}=0$
for all $k\neq0$ (similar conventions are used to $\beta_{ijk}^{\left(r,s\right)}$),
in other words if the gene locus is present on a gonosome it can not
disappear and if absent it can not appear. It follows from this that
the progeny of two individuals, one of genotype $a_{i}^{\left(r\right)}a_{j}^{\left(s\right)}$
and the other of genotype $a_{p}^{\left(u\right)}a_{q}^{\left(v\right)}$
is:
\[
\tfrac{1}{4}\alpha_{ijk}^{\left(r,s\right)}\alpha_{pql}^{\left(u,v\right)}a_{k}^{\left(r\right)}a_{l}^{\left(u\right)},\tfrac{1}{4}\alpha_{ijk}^{\left(r,s\right)}\beta_{pql}^{\left(u,v\right)}a_{k}^{\left(r\right)}a_{l}^{\left(v\right)},\tfrac{1}{4}\beta_{ijk}^{\left(r,s\right)}\alpha_{pql}^{\left(u,v\right)}a_{k}^{\left(s\right)}a_{l}^{\left(u\right)},\tfrac{1}{4}\beta_{ijk}^{\left(r,s\right)}\beta_{pql}^{\left(u,v\right)}a_{k}^{\left(s\right)}a_{l}^{\left(v\right)}.
\]

We can associate with this situation a gonosomal algebra by setting
in proposition \ref{pro:AGetDupliquee}, the vector space $A$ with
basis $\left\{ a_{i}^{\left(r\right)};1\leq r\leq3,i\in I_{r}\right\} $
and commutative multiplication $a_{i}^{\left(r\right)}a_{j}^{\left(s\right)}=\frac{1}{2}\sum_{k\in I_{r}}\alpha_{ijk}^{\left(r,s\right)}a_{k}^{\left(r\right)}+\frac{1}{2}\sum_{k\in I_{s}}\beta_{ijk}^{\left(r,s\right)}a_{k}^{\left(s\right)},$
the algebra $A$ is weighted by $\omega\left(a_{i}^{\left(r\right)}\right)=1$.
Then we take the following subspaces of $D\left(A\right)$ :
\begin{eqnarray*}
A_{1} & = & span\left(\left\{ a_{i_{1}}^{\left(1\right)}\otimes a_{i_{2}}^{\left(2\right)},a_{i_{1}}^{\left(1\right)}\otimes a_{i_{3}}^{\left(3\right)}\right\} _{i_{1}\in I_{1,},i_{2}\in I_{2},i_{3}\in I_{3}}\cup\left\{ a_{i_{2}}^{\left(2\right)}\otimes a_{j_{2}}^{\left(2\right)};i_{2}\leq j_{2}\right\} _{i_{2},j_{2}\in I_{2}}\right),\\
A_{2} & = & span\left(\left\{ a_{i_{2}}^{\left(2\right)}\otimes a_{i_{3}}^{\left(3\right)},a_{i_{3}}^{\left(3\right)}\otimes a_{j_{3}}^{\left(3\right)}\right\} _{i_{2}\in I_{2},i_{3},j_{3}\in I_{3}}\right).
\end{eqnarray*}

\end{example}
\medskip{}

\begin{example}
\emph{Heredity in the WZ-system with male feminization.}

By 1940, the French biologist Albert Vandel noted that some female
woodlices of the species \emph{Armadillidium vulgare} generate 80\%
to 100\% of females instead of the expected 50\% and that this seems
hereditary. The explanation of this phenomenon was discovered in 1973:
these females with almost exclusively female descent are in reality
males infected by a bacterium which transforms them into females.
This bacterium was identified in 1992, it belongs to the genus \emph{Wolbachia},
it is an endosymbiotic bacterial cell widespread in insects which
mainly transmitted from mother to offspring and changes the reproduction
of its guests either by making sterile the matings, or by killing
embryos, or still by feminizing males.

Biologically, the sex of woodlices follows the $WZ$-system. When
a male $ZZ$ is infected by Wolbachia, what is denoted $ZZ+w$, it
becomes female and can cross with a male $ZZ$, and as transmission
of Wolbachia does not happen 100 \% there is a majority of $ZZ+w$
females and a minority of $ZZ$ males. Finally, in this population
three kinds of females is observed: $WZ$, $WZ+w$, $ZZ+w$. For the
crosses, if we note $\eta$ ($0.5<\eta<1$) the transmission rate
of Wolbachia in the offspring, we have: \medskip{}

\hspace{1cm}%
\begin{tabular}{rl}
$WZ\times ZZ$ & $\rightarrowtail\frac{1}{2}WZ,\frac{1}{2}ZZ;$\tabularnewline
$\left(WZ+w\right)\times ZZ$ & $\rightarrowtail\frac{\eta}{2}\left(WZ+w\right),\frac{\eta}{2}\left(ZZ+w\right),\frac{\left(1-\eta\right)}{2}WZ,\frac{\left(1-\eta\right)}{2}ZZ;$\tabularnewline
$\left(ZZ+w\right)\times ZZ$ & $\rightarrowtail\eta\left(ZZ+w\right),\left(1-\eta\right)ZZ.$\tabularnewline
\end{tabular}

Algebraically, we take a vector space $A$ with basis $\left(e_{1},\ldots,e_{4}\right)$
equipped with the multiplication: \smallskip{}

\hspace{2cm}%
\begin{tabular}{l}
$e_{i}^{2}=e_{i}$, ($i=1,\ldots,4$)\tabularnewline
$e_{1}e_{2}=\left(1-\eta\right)e_{1}+\eta e_{2}$,\tabularnewline
$e_{1}e_{3}=\frac{1}{2}e_{1}+\frac{1}{2}e_{3}$,\tabularnewline
$e_{1}e_{4}=e_{2}e_{3}=e_{2}e_{4}=\frac{\left(1-\eta\right)}{2}\left(e_{1}+e_{3}\right)+\frac{\eta}{2}\left(e_{2}+e_{4}\right)$,\tabularnewline
$e_{3}e_{4}=\frac{1}{2}e_{3}+\frac{1}{2}e_{4}$.\tabularnewline
\end{tabular}

\smallskip{}

The algebra $A$ is weighted by $\omega\left(e_{i}\right)=1$. We
take in the proposition \ref{pro:AGetDupliquee}, the subspaces $A_{1}=\mathbb{R}\left\langle e_{1}\otimes e_{1}\right\rangle $
and $A_{2}=\mathbb{R}\left\langle e_{1}\otimes e_{2},e_{1}\otimes e_{3},e_{1}\otimes e_{4}\right\rangle $,
if in the products obtained from the proposition \ref{pro:AGetDupliquee}
we make the correspondences $e_{1}\leftrightarrow Z$, $e_{2}\leftrightarrow Z+w$,
$e_{3}\leftrightarrow W$ and $e_{4}\leftrightarrow W+w$, we retrieve
the results of crosses given in the biological model.
\end{example}
\smallskip{}

\begin{example}
\emph{Heredity in the $XY$-system with fertile $XY$ females.}

In rodents, cases where the $XY$ system is atypical have been found.
For example, in \emph{Myopus schisticolor} (wood lemming) and \emph{Mus
minutoides} (african pygmy mouse), we described three female genotypes:
$XX$, $XX^{*}$ and $X^{*}Y$. In these genotypes, notation $X^{*}$
refers to a chromosome carrying a gene having two actions: it inactivates
the action of gonosome $Y$ and it causes the elimination of gonosome
$Y$ during gametogenesis, so $X^{*}Y$ females give only ova of $X^{*}$
type. Therefore the results of crosses are:

\smallskip{}

\hspace{3cm}%
\begin{tabular}{l}
$XX\times XY\rightarrowtail\frac{1}{2}XX,\frac{1}{2}XY;$\tabularnewline
$XX^{*}\times XY\rightarrowtail\frac{1}{4}XX,\frac{1}{4}XX^{*},\frac{1}{4}X^{*}Y,\frac{1}{4}XY;$\tabularnewline
\multicolumn{1}{l}{$X^{*}Y\times XY\rightarrowtail\frac{1}{2}XX^{*},\frac{1}{2}X^{*}Y.$}\tabularnewline
\end{tabular}

\smallskip{}
To interpret this in algebraic terms, we define the space $A$ with
basis $\left(e_{1},e_{2},e_{3}\right)$ and the multiplication:
\begin{eqnarray*}
e_{i}^{2}=e_{i}, & e_{1}e_{i}=\frac{1}{2}e_{1}+\frac{1}{2}e_{i}, & e_{2}e_{3}=e_{2},\quad\left(i=1,\ldots,3\right).
\end{eqnarray*}
Next we take in proposition \ref{pro:AGetDupliquee}, the subpaces
$A_{1}=\mathbb{R}\left\langle e_{1}\otimes e_{1},e_{1}\otimes e_{2},e_{2}\otimes e_{3}\right\rangle $
and $A_{2}=\mathbb{R}\left\langle e_{1}\otimes e_{3}\right\rangle $,
then the product defined in proposition \ref{pro:AGetDupliquee} and
the relations $e_{1}\leftrightarrow X$, $e_{2}\leftrightarrow X^{*}$,
$e_{3}\leftrightarrow Y$, allow to find the results of crosses.

\medskip{}

In \emph{Dicrostonyx torquatus} (Arctic lemming) three female genotypes
are also described : $XX$, $XX^{*}$ and $X^{*}Y$, but unlike the
previous case, the females $X^{*}Y$ give normal male $XY$. In this
case the results of crosses are:

\smallskip{}

\hspace{3cm}%
\begin{tabular}{l}
$XX\times XY\rightarrowtail\frac{1}{2}XX,\frac{1}{2}XY;$\tabularnewline
$XX^{*}\times XY\rightarrowtail\frac{1}{4}XX,\frac{1}{4}XX^{*},\frac{1}{4}X^{*}Y,\frac{1}{4}XY;$\tabularnewline
\multicolumn{1}{l}{$X^{*}Y\times XY\rightarrowtail\frac{1}{3}XX^{*},\frac{1}{3}X^{*}Y,\frac{1}{3}XY.$}\tabularnewline
\end{tabular}

\smallskip{}
And in this case the algebraic model is defined from the algebra $A=\mathbb{R}\left\langle e_{1},e_{2},e_{3}\right\rangle $
where 
\begin{eqnarray*}
e_{i}^{2}=e_{i}, & e_{1}e_{i}=\frac{1}{2}e_{1}+\frac{1}{2}e_{i}, & e_{2}e_{3}=e_{2},\quad\left(i=1,\ldots,3\right).
\end{eqnarray*}

Then we take in proposition \ref{pro:AGetDupliquee}, the subspaces
$A_{1}=\mathbb{R}\left\langle e_{1}\otimes e_{1},e_{1}\otimes e_{2},e_{2}\otimes e_{3}\right\rangle $
and $A_{2}=\mathbb{R}\left\langle e_{1}\otimes e_{3},e_{3}\otimes e_{3}\right\rangle $,
next we apply proposition \ref{pro:Reductionbase} to obtain a gonosomal
algebra structure on $A_{1}\oplus\mathbb{R}\left\langle e_{1}\otimes e_{3}\right\rangle $
from which relations $e_{1}\leftrightarrow X$, $e_{2}\leftrightarrow X^{*}$,
$e_{3}\leftrightarrow Y$ allow to find the results of crosses.
\end{example}
\smallskip{}

\subsection{Construction from two baric algebras.}
\begin{prop}
\label{pro:AGetAlgPond} Let $\left(A,\omega\right)$ and $(\widetilde{A},\widetilde{\omega})$
be (not necessarily commutative) finite-dimensional baric $K$-algebras,
the Etherington morphism $\mu:A\otimes A\rightarrow A^{2}$, $a\otimes a'\mapsto aa'$
and $\varphi:A\otimes\widetilde{A}\rightarrow A$, $\widetilde{\varphi}:A\otimes\widetilde{A}\rightarrow\widetilde{A}$
two linear maps such that $\omega\circ\varphi+\widetilde{\omega}\circ\widetilde{\varphi}=\omega\otimes\widetilde{\omega}$,
then the $K$-space $A\otimes A\times A\otimes\widetilde{A}$ equipped
with the algebra structure:
\[
\left(x,y\right)\left(x',y'\right)=\left(\mu\left(x\right)\otimes\varphi\left(y'\right)+\mu\left(x'\right)\otimes\varphi\left(y\right),\mu\left(x\right)\otimes\widetilde{\varphi}\left(y'\right)+\mu\left(x'\right)\otimes\widetilde{\varphi}\left(y\right)\right)
\]
 is a gonosomal algebra.\end{prop}
\begin{proof}
For all $x\in A\otimes A$ and $y'\in A\otimes\widetilde{A}$ identifying
$\left(x,0\right)$ to $x$ and $\left(0,y'\right)$ to $y'$, the
multiplication given in the statement becomes: 
\[
xy'=\mu\left(x\right)\otimes\varphi\left(y'\right)\oplus\mu\left(x\right)\otimes\widetilde{\varphi}\left(y'\right).\quad\left(*\right)
\]

So for all $x,x'\in A\otimes A$ and $y,y'\in A\otimes\widetilde{A}$
we have $y'x=xy'$, $xx'=0$ and $yy'=0$. Algebras $A$ and $\widetilde{A}$
being finite-dimensional weighted there is a basis $\left(a_{i}\right)_{1\text{\ensuremath{\le}}i\text{\ensuremath{\le}}n}$
of $A$ and a basis $\left(\widetilde{a}_{p}\right)_{1\text{\ensuremath{\le}}p\text{\ensuremath{\le}}m}$
of $\widetilde{A}$ such that $\omega\left(a_{i}\right)=1$ and $\widetilde{\omega}\left(\widetilde{a}_{p}\right)=1$.
Let $\tau$ be a bijection from $\left\{ 1,\text{\dots},n\right\} ^{2}$
to $\left\{ 1,\text{\dots},n^{2}\right\} $ ordering the basis $\left(a_{i}\otimes a_{j}\right)_{i,j}$
of $A\otimes A$, we put $e_{\tau\left(i,j\right)}=\left(a_{i}\otimes a_{j},0\right)$,
$\left(1\leq i,j\leq n\right)$. Let $\sigma$ be a bijection from
$\left\{ 1,\text{\dots},n\right\} \times\left\{ 1,\text{\dots},m\right\} $
to $\left\{ 1,\text{\dots},nm\right\} $ ordering the basis $\left(a_{i}\otimes\widetilde{a}_{p}\right)_{i,p}$
of $A\otimes\widetilde{A}$, we put $\widetilde{e}_{\sigma\left(i,p\right)}=\left(0,a_{i}\otimes\widetilde{a}_{p}\right)$.
Next for all $1\text{\ensuremath{\le}}i,j\text{\ensuremath{\le}}n$
and $1\text{\ensuremath{\le}}p\text{\ensuremath{\le}}m$ let $\varphi\left(\widetilde{e}_{\sigma\left(i,p\right)}\right)=\sum_{k=1}^{n}\alpha_{k,\sigma\left(i,p\right)}a_{k}$,
$\widetilde{\varphi}\left(\widetilde{e}_{\sigma\left(i,p\right)}\right)=\sum_{s=1}^{m}\beta_{s,\sigma\left(i,p\right)}\widetilde{a}_{s}$
and $\mu\left(a_{i}\otimes a_{j}\right)=a_{i}a_{j}=\sum_{k=1}^{n}\lambda_{ijk}a_{k}$.
With this the identity $\left(*\right)$ is written:
\begin{eqnarray*}
e_{\tau\left(i,j\right)}\widetilde{e}_{\sigma\left(p,q\right)} & = & \left(a_{i}\otimes a_{j},0\right)\left(0,a_{p}\otimes\widetilde{a}_{q}\right)\\
 & = & \sum_{k,r=1}^{n}\lambda_{ijk}\alpha_{r,\sigma\left(p,q\right)}e_{\tau\left(k,r\right)}\oplus\sum_{k=1}^{n}\sum_{s=1}^{m}\lambda_{ijk}\beta_{s,\sigma\left(p,q\right)}\widetilde{e}_{\sigma\left(k,s\right)}.\quad\left(**\right)
\end{eqnarray*}
So noting:
\[
\gamma_{\tau\left(i,j\right),\sigma\left(p,q\right),\tau\left(k,r\right)}=\lambda_{ijk}\alpha_{r,\sigma\left(p,q\right)}\mbox{ and }\widetilde{\gamma}_{\tau\left(i,j\right),\sigma\left(p,q\right),\sigma\left(k,s\right)}=\lambda_{ijk}\beta_{s,\sigma\left(p,q\right)},
\]
the identity $\left(**\right)$ becomes:
\[
e_{\tau\left(i,j\right)}\widetilde{e}_{\sigma\left(p,q\right)}=\sum_{\tau\left(k,r\right)=1}^{n^{2}}\gamma_{\tau\left(i,j\right),\sigma\left(p,q\right),\tau\left(k,r\right)}e_{\tau\left(k,r\right)}+\sum_{\sigma\left(k,s\right)=1}^{nm}\widetilde{\gamma}_{\tau\left(i,j\right),\sigma\left(p,q\right),\sigma\left(k,s\right)}\widetilde{e}_{\sigma\left(k,s\right)}.
\]
Now, we have:
\begin{eqnarray*}
\sum_{k=1}^{n}\alpha_{k,\sigma\left(i,p\right)}+\sum_{s=1}^{m}\beta_{s,\sigma\left(i,p\right)} & = & \omega\left(\varphi\left(a_{i}\otimes\widetilde{a}_{p}\right)\right)+\widetilde{\omega}\left(\widetilde{\varphi}\left(a_{i}\otimes\widetilde{a}_{p}\right)\right)\\
 & = & \omega\otimes\widetilde{\omega}\left(a_{i}\otimes\widetilde{a}_{p}\right)=\omega\left(a_{i}\right)\widetilde{\omega}\left(\widetilde{a}_{p}\right)=1
\end{eqnarray*}
and $\sum_{k=1}^{n}\lambda_{ijk}=\omega\left(a_{i}a_{j}\right)=1$,
we deduce that 
\[
\sum_{\tau\left(k,r\right)}\gamma_{\tau\left(i,j\right),\sigma\left(p,q\right),\tau\left(k,r\right)}+\sum_{\sigma\left(k,s\right)}\widetilde{\gamma}_{\tau\left(i,j\right),\sigma\left(p,q\right),\sigma\left(k,s\right)}=1,\quad(1\text{\ensuremath{\le}}i,j\text{\ensuremath{\le}}n,1\text{\ensuremath{\le}}p,q\text{\ensuremath{\le}}m),
\]
what establishes that $A\otimes A\times A\otimes\widetilde{A}$ is
a gonosomal algebra.\end{proof}
\begin{rem}
If in the above proposition we take $\left(A,\omega\right)$ a baric
algebra, $\widetilde{A}$ the $K$-algebra spanned by an element $Y$
verifying $Y^{2}=Y$ weighted by $\widetilde{\omega}\left(Y\right)=1$,
the maps $\varphi:A\otimes\widetilde{A}\rightarrow A$, $\varphi\left(x\otimes Y\right)=\frac{1}{2}x$
and $\widetilde{\varphi}:A\otimes\widetilde{A}\rightarrow\widetilde{A}$,
$\widetilde{\varphi}\left(x\otimes Y\right)=\frac{1}{2}\omega\left(x\right)Y$.
Then we have:
\[
\left(x\otimes y\oplus z\otimes Y\right)\left(x'\otimes y'\oplus z'\otimes Y\right)=\frac{1}{2}\left(xy\otimes z'+x'y'\otimes z\right)+\frac{1}{2}\left(\omega\left(z'\right)xy\otimes Y+\omega\left(z\right)x'y'\otimes Y\right)
\]
and after identification of $A\otimes\widetilde{A}$ with $A$ we
find the law (\ref{eq:Gonshor}).\end{rem}
\begin{example}
\emph{X-linked sex ratio distorter.} 

In some species the $X$ chromosome carries alleles referred as distorters
which disrupt in male the production of gametes carrying $Y$ chromosome.
In this case the males offspring consists only of females. However
in some species suppressors of distortion are observed on the $Y$
chromosome, in this case depending on the suppressor type the percentage
of females in the male offspring is between 50\% ($Y$ totally suppressor)
and 100\% ($Y$ no suppressor).

If we denote $X^{d}$ the $X$ chromosome carrying an distorter allele
and $Y^{s}$ the $Y$ chromosome carrying a suppressor gene, there
are three female genotypes $XX$, $XX^{d}$ , $X^{d}X^{d}$ and four
male genotypes $XY$, $XY^{s}$, $X^{d}Y$, $X^{d}Y^{s}$ and gametogenesis
of these genotypes are $XY\rightarrowtail\frac{1}{2}X+\frac{1}{2}Y$,
$XY^{s}\rightarrowtail\frac{1}{2}X+\frac{1}{2}Y^{s}$, $X^{d}Y\rightarrowtail X^{d}$,
$X^{d}Y^{s}\rightarrowtail\frac{2-\theta}{2}X^{d}+\frac{\theta}{2}Y^{s}$
where $0\leq\theta\leq1$ is the suppression rate of distortion. For
example, the cross $XX\times X^{d}Y$ gives only $XX^{d}$ females
and $XX\times X^{d}Y^{s}\rightarrowtail\frac{2-\theta}{2}XX^{d},\frac{\theta}{2}XY^{s}$
.

To show that this situation is depicted by a gonosomal algebra we
use the correspondences $X\leftrightarrow e_{1}$, $X^{d}\leftrightarrow e_{2}$,
$Y\leftrightarrow\widetilde{e}_{1}$ and $Y^{s}\leftrightarrow\widetilde{e}_{2}$
and we apply proposition \ref{pro:AGetAlgPond} with $A=\mathbb{R}\left\langle e_{1},e_{2}\right\rangle $,
$\widetilde{A}=\mathbb{R}\left\langle \widetilde{e}_{1},\widetilde{e}_{2}\right\rangle $
equipped with the algebraic structures: $e_{i}e_{j}=\frac{1}{2}e_{i}+\frac{1}{2}e_{j}$
and $\widetilde{e}_{i}\widetilde{e}_{j}=\frac{1}{2}\widetilde{e}_{i}+\frac{1}{2}\widetilde{e}_{j}$,
then $A$ and $\widetilde{A}$ are weighted by $\omega\left(e_{i}\right)=1$
and $\widetilde{\omega}\left(\widetilde{e}_{i}\right)=1$. We define
$\varphi:A\otimes\widetilde{A}\rightarrow A$ and $\widetilde{\varphi}:A\otimes\widetilde{A}\rightarrow\widetilde{A}$
by 
\[
\varphi\left(e_{i}\otimes\widetilde{e}_{j}\right)=\begin{cases}
\tfrac{1}{2}e_{1} & \mbox{si \ensuremath{i=1}}\\
e_{2} & \mbox{si }\ensuremath{\left(i,j\right)=\left(2,1\right)}\\
\tfrac{2-\theta}{2}e_{2} & \mbox{si }\ensuremath{\left(i,j\right)=\left(2,2\right)}
\end{cases},\quad\widetilde{\varphi}\left(e_{i}\otimes\widetilde{e}_{j}\right)=\begin{cases}
\tfrac{1}{2}\widetilde{e}_{1} & \mbox{si \ensuremath{i=1}}\\
0 & \mbox{si }\ensuremath{\left(i,j\right)=\left(2,1\right)}\\
\tfrac{\theta}{2}\widetilde{e}_{2} & \mbox{si }\ensuremath{\left(i,j\right)=\left(2,2\right).}
\end{cases}
\]

\end{example}
\smallskip{}

\begin{example}
\emph{Kleptogenesis}.

Kleptogenesis is observed in fishes (\emph{Poeciliopsis}), frogs (\emph{Pelophylax})
and insects (\emph{Bacillus}). These bisexual species are present
in various forms that can hybridise, that is to say cross-fertilizations
is observed between two related species. During oogenesis hybrid females
produce eggs that contain only the genome of one of his parents transmitted
without recombination, the genome of the other parental species being
totally evacuated. This mode of reproduction is thus characterized
by a clonal transmission of the genome, whereas the other genome is
acquired by sexual reproduction.

The first observation of this phenomenon was made by L. Berger in
three species of European green frogs (genus \emph{Pelophylax}): \emph{P.
lessonae} (LL) small, \emph{P. ridibundus} (RR) large and \emph{P.
esculentus} (LR ) whose size is intermediate. If frogs (LL) and (RR)
are crossed, hybrids (LR) are obtained, but if we cross a female (LR)
with a male (LL) only (LR) are obtained while the Mendel law provides
50\% (LL) and 50\% (LR). This result is explained by the systematic
elimination during the (LR) frogs gametogenesis of the entire genome
inherited from parents (LL).

Algebraically, we apply the proposition \ref{pro:AGetAlgPond} by
taking $A=\widetilde{A}$ an algebra with basis $\left(a_{1},a_{2}\right)$
and multiplication $a_{i}a_{j}=\frac{1}{2}\left(a_{i}+a_{j}\right)$
weighted by $\omega\left(a_{i}\right)=1$. Let $0<\theta<1$ be the
female proportion in the population, we define $\varphi,\widetilde{\varphi}:A\otimes\widetilde{A}\rightarrow A$
by 
\begin{eqnarray*}
\varphi\left(a_{i}\otimes a_{j}\right) & = & \begin{cases}
\left(1-\theta\right)a_{i} & \mbox{si \ensuremath{i=j}}\\
\left(1-\theta\right)a_{1} & \mbox{si }\ensuremath{\left(i,j\right)=\left(1,2\right),\left(2,1\right),}
\end{cases}\\
\widetilde{\varphi}\left(a_{i}\otimes a_{j}\right) & = & \begin{cases}
\theta a_{i} & \mbox{si \ensuremath{i=j}}\\
\theta a_{1} & \mbox{si }\ensuremath{\left(i,j\right)=\left(1,2\right),\left(2,1\right).}
\end{cases}
\end{eqnarray*}

Noting $e_{ij}=\left(0,a_{i}\otimes a_{j}\right)$ female genotypes
and $\widetilde{e}_{ij}=\left(a_{i}\otimes a_{j},0\right)$ male genotypes,
the product defined in the proposition \ref{pro:AGetAlgPond} coupled
to the relations $a_{1}\leftrightarrow R$, $a_{2}\leftrightarrow L$
give the results of crosses observed in species \emph{Pelophylax.}
\end{example}
\smallskip{}

The following three examples show that this construction is also useful
to give algebraic models of genetic processes influenced by sex.
\begin{example}
\emph{Recombination between two pseudo-autosomal genes}.

This case was investigated for zygotic algebra in \cite{DA-LS-PV-89}.
We consider two pseudo-autosomal genes $a$ and $b$ in the $XY$
system which allelic forms are $a_{1},\text{\dots},a_{n}$ and $b_{1},\text{\dots},b_{m}$,
it is also assumed that the locus of the gene $a$ is closer to the
centromere and does not occur recombination between the locus and
the centromere, it follows that the alleles of $a$ does not change
chromosome during recombination.

We note $\left(a_{i}b_{j},a_{k}b_{l}\right)$ the genotype of an individual
where the haplotype $a_{i}b_{j}$ is transmitted by the mother and
$a_{k}b_{l}$ by the father. Let $\theta$ and $\widetilde{\theta}$
respectively the recombination rates between these two loci in females
and males. During meiosis a female $\left(a_{i}b_{j},a_{k}b_{l}\right)$
produces ova in the following proportions: $\left(\frac{1-\theta}{2}a_{i}b_{j},\frac{1-\theta}{2}a_{k}b_{l},\frac{\theta}{2}a_{i}b_{l},\frac{\theta}{2}a_{k}b_{j}\right)$.
In a male $\left(a_{p}b_{q},a_{r}b_{s}\right)$, the distribution
of spermatozoa is $\left(\frac{1-\widetilde{\theta}}{2}a_{p}b_{q},\frac{1-\widetilde{\theta}}{2}a_{r}b_{s},\frac{\widetilde{\theta}}{2}a_{p}b_{s},\frac{\widetilde{\theta}}{2}a_{r}b_{q}\right)$
where, taking into consideration the notation, genotypes $a_{p}b_{q}$
and $a_{p}b_{s}$ are carried by a gonosome $X$ and $a_{r}b_{s}$,
$a_{r}b_{q}$ are on a gonosome $Y$. It results for example, that
the frequency distribution of the eight genotypes of girls born from
the cross between a female $\left(a_{i}b_{j},a_{k}b_{l}\right)$ and
a male $\left(a_{p}b_{q},a_{r}b_{s}\right)$ is:

$\frac{\left(1-\theta\right)\widetilde{\theta}}{4}\left(a_{i}b_{j},a_{p}b_{s}\right)$,
$\frac{\left(1-\theta\right)\widetilde{\theta}}{4}\left(a_{i}b_{j},a_{p}b_{s}\right)$,
$\frac{\left(1-\theta\right)\left(1-\widetilde{\theta}\right)}{4}\left(a_{k}b_{l},a_{p}b_{q}\right)$,
$\frac{\left(1-\theta\right)\widetilde{\theta}}{4}\left(a_{k}b_{l},a_{p}b_{s}\right)$, 

$\frac{\theta\left(1-\widetilde{\theta}\right)}{4}\left(a_{i}b_{l},a_{p}b_{q}\right)$,
$\frac{\theta\widetilde{\theta}}{4}\left(a_{i}b_{l},a_{p}b_{s}\right)$,
$\frac{\theta\left(1-\widetilde{\theta}\right)}{4}\left(a_{k}b_{j},a_{p}b_{q}\right)$,
$\frac{\theta\widetilde{\theta}}{4}\left(a_{k}b_{j},a_{p}b_{s}\right)$.

With proposition \ref{pro:AGetAlgPond} we show that this situation
is modeled by a gonosomal algebra. Indeed, let $A$, $\widetilde{A}$
be the vector spaces spanned by $\left(a_{i}\otimes b_{j}\right)$
respectively equipped with algebraic structures: 
\begin{eqnarray*}
\left(a_{i}\otimes b_{j}\right)\left(a_{k}\otimes b_{l}\right) & = & \tfrac{1-\theta}{2}\left(a_{i}\otimes b_{j}+a_{k}\otimes b_{l}\right)+\tfrac{\theta}{2}\left(a_{i}\otimes b_{l}+a_{k}\otimes b_{j}\right),\\
\left(a_{i}\otimes b_{j}\right)*\left(a_{k}\otimes b_{l}\right) & = & \tfrac{1-\widetilde{\theta}}{2}\left(a_{i}\otimes b_{j}+a_{k}\otimes b_{l}\right)+\tfrac{\widetilde{\theta}}{2}\left(a_{i}\otimes b_{l}+a_{k}\otimes b_{j}\right).
\end{eqnarray*}
Then $A$ is weighted by $\omega\left(a_{i}\otimes a_{j}\right)=1$
and $\widetilde{A}$ by $\widetilde{\omega}\left(a_{i}\otimes a_{j}\right)=1$.
We define the maps $\varphi:A\otimes\widetilde{A}\rightarrow A$ by
$\varphi\left(\left(a_{i}\otimes b_{j}\right)\otimes\left(a_{k}\otimes b_{l}\right)\right)=\frac{\left(1-\widetilde{\theta}\right)}{2}a_{i}\otimes b_{j}+\frac{\widetilde{\theta}}{2}a_{i}\otimes b_{l}$
and $\widetilde{\varphi}:A\otimes\widetilde{A}\rightarrow\widetilde{A}$
by $\widetilde{\varphi}\left(\left(a_{i}\otimes b_{j}\right)\otimes\left(a_{k}\otimes b_{l}\right)\right)=\frac{\left(1-\widetilde{\theta}\right)}{2}a_{k}\otimes b_{l}+\frac{\widetilde{\theta}}{2}a_{k}\otimes b_{j}$,
we have $\omega\circ\varphi+\widetilde{\omega}\circ\widetilde{\varphi}=\omega\otimes\widetilde{\omega}$.
\end{example}
\smallskip{}

\begin{example}
\emph{Genetic mutation with sex-dependent mutation rate in a multiple
$XY$-system.} 

The inheritance of a gonosomal gene was studied for the $XY$ system
in the absence of mutation in \cite{Ether-41,Gonsh-73,Holg-70} and
with mutation in \cite{Gonsh-60}. Nevertheless it is observed in
all sex determination systems, in the following it is described in
the multiple $XY$-system for pseudo-autosomal or gonosomal genes.

There are many diploid species (fish, insects, spiders) who have sex
determination systems with multiple sex chromosomes. In these systems
sex is determined by the $X_{n}Y_{m}$ system with $n\geq1$ and $m\geq0$,
that is to say a female type is $X_{1}\ldots X_{n}/X_{1}\ldots X_{n}$
and a male $X_{1}\ldots X_{n}/Y_{1}\ldots Y_{m}$ if $m\geq1$, or
$X_{1}\ldots X_{n}/0$ if $m=0$ ($W_{m}Z_{n}$ system is symmetric).

We consider in this population, a gonosomal gene whose allelic forms
are $a_{1},\ldots,a_{N}$ and we denote by $a_{0}$ the case where
the locus of this gene is not observed on a gonosome. We use the following
notations, for every $1\leq k\leq n$ (resp. $1\leq r\leq m$) we
put $I_{k}=\left\{ 1,\ldots,N\right\} $ (resp. $J_{r}=\left\{ 1,\ldots,N\right\} $)
if the chromosome $X_{k}$ (resp. $Y_{r}$) carries the gene and $I_{k}=\left\{ 0\right\} $
(resp. $J_{r}=\left\{ 0\right\} $) otherwise. If we denote by $\binom{i_{1,1},\ldots,i_{1,n}}{i_{2,1},\ldots,i_{2,n}}$
where $\left(i_{1,1},\ldots,i_{1,n}\right)$, $\left(i_{2,1},\ldots,i_{2,n}\right)\in\prod_{k=1}^{n}I_{k}$
a female genotype, and by $\binom{p_{1},\ldots,p_{n}}{q_{1},\ldots,q_{m}}$
with $\left(p_{1},\ldots,p_{n}\right)\in\prod_{k=1}^{n}I_{k}$, $\left(q_{1},\ldots,q_{m}\right)\in\prod_{r=1}^{m}J_{r}$
a male genotype, then in the absence of mutation and genetic recombination
and taking into account that the segregation of the homologous $X$
chromosomes is independent, the progeny of a female $\binom{i_{1,1},\ldots,i_{1,n}}{i_{2,1},\ldots,i_{2,n}}$
with a male $\binom{p_{1},\ldots,p_{n}}{q_{1},\ldots,q_{m}}$ consists
of $\left(\frac{1}{2}\right)^{n+1}$females of genotype $\binom{i_{s_{1},1},\ldots,i_{s_{n},n}}{p_{1},\ldots,p_{n}}$
and $\left(\frac{1}{2}\right)^{n+1}$males of genotype $\binom{i_{s_{1},1},\ldots,i_{s_{n},n}}{q_{1},\ldots,q_{m}}$,
where $s_{1},\ldots,s_{n}\in\left\{ 1,2\right\} $.

In the presence of mutations whose rates depend on the sex of individuals,
when the chromosome $X_{k}$ carries the gene, we note $\mu_{k,ri}$
(resp. $\widetilde{\mu}_{k,ri}$) the mutation rate of the allele
$a_{i}$ to allele $a_{r}$ on the $X_{k}$ chromosome in female (resp.
male), we have $\sum_{r\in I_{k}}\mu_{k,ri}=\sum_{r\in I_{k}}\widetilde{\mu}_{k,ri}=1$
and by convention we put $\mu_{k,0i}=0$. In a similar way, when the
locus of the gene is observed on chromosome $Y_{k}$, we note $\nu_{k,ri}$
the mutation rate of allele $a_{i}$ to allele $a_{r}$ on chromosome
$Y_{k}$ with the convention $\nu_{k,0i}=0$. If the gene locus is
not on gonosome $X_{l}$ (resp. $Y_{l}$) we put $\mu_{l,00}=\widetilde{\mu}_{l,00}=1$
(resp. $\nu_{l,00}=1$) and thus $\mu_{l,k0}=\widetilde{\mu}_{l,k0}=\nu_{l,k0}=0$
for $k\neq0$.

With the previous notations, assuming that mutations between chromosomes
are independent and in the absence of recombination, after crossing
a female of genotype $\binom{i_{1,1},\ldots,i_{1,n}}{i_{2,1},\ldots,i_{2,n}}$
with a male of genotype $\binom{p_{1},\ldots,p_{n}}{q_{1},\ldots,q_{m}}$
we obtain the proportions:

$\quad\left(\frac{1}{2}\right)^{n+1}\prod_{k=1}^{n}\left(\mu_{k,r_{k}i_{1,k}}+\mu_{k,r_{k}i_{2,k}}\right)\widetilde{\mu}_{k,t_{k}p_{k}}$
females of genotype $\binom{r_{1},\ldots,r_{n}}{t_{1},\ldots,t_{n}}$ 

and $\left(\frac{1}{2}\right)^{n+1}\prod_{k=1}^{n}\left(\mu_{k,r_{k}i_{1,k}}+\mu_{k,r_{k}i_{2,k}}\right)\prod_{k=1}^{m}\nu_{k,t_{k}q_{k}}$
males of genotype $\binom{r_{1},\ldots,r_{n}}{t_{1},\ldots,t_{m}}$. 

\smallskip{}

To interpret this model in gonosomal algebra terms, we apply the proposition
\ref{pro:AGetAlgPond}. For $1\leq k\leq n$, we put $A_{k}=\mbox{span}\left\{ a_{i};i\in I_{k}\right\} $
and $A=\bigotimes_{k=1}^{n}A_{k}$. When $m\geq1$, for $1\leq r\leq m$
we note $\widetilde{A}_{r}=\mbox{span}\left\{ a_{i};i\in J_{r}\right\} $
and $\widetilde{A}=\bigotimes_{r=1}^{m}\widetilde{A}_{r}$; when $m=0$
we put $\widetilde{A}=\mbox{span}\left\{ a_{0}\right\} $. First we
provide every space $A_{k}$ with the algebra structure $a_{i}a_{j}=\frac{1}{2}\sum_{r\in I_{k}}\left(\mu_{k,ri}+\mu_{k,rj}\right)a_{r}$
and the weight function $\omega_{k}\left(a_{i}\right)=1$ for all
$i,j\in I_{k}$, and every space $\widetilde{A}_{k}$ with the multiplication
$a_{i}a_{j}=\frac{1}{2}\sum_{r\in J_{k}}\left(\widetilde{\mu}_{k,ri}+\widetilde{\mu}_{k,rj}\right)a_{r}$
and the weight function $\widetilde{\omega}_{k}\left(a_{i}\right)=1$.
Next we equip the spaces $A$ and $\widetilde{A}$ with the algebraic
structure $\left(\bigotimes_{k}x_{k}\right)\left(\bigotimes_{k}y_{k}\right)=\bigotimes_{k}\left(x_{k}y_{k}\right)$,
the weight functions $\omega=\bigotimes_{k=1}^{n}\omega_{k}$ for
$A$ and $\widetilde{\omega}=\bigotimes_{k=1}^{m}\widetilde{\omega}_{k}$
for $\widetilde{A}$. Finally we define the maps $\varphi:A\otimes\widetilde{A}\rightarrow A$
and $\widetilde{\varphi}:A\otimes\widetilde{A}\rightarrow\widetilde{A}$
by $\varphi\left(\bigotimes_{k=1}^{n}a_{i_{k}}\otimes\bigotimes_{k=1}^{m}a_{j_{k}}\right)=\bigotimes_{k=1}^{n}\left(\sum_{r\in I_{k}}\widetilde{\mu}_{k,ri_{k}}a_{r}\right)$
and $\widetilde{\varphi}\left(\bigotimes_{k=1}^{n}a_{i_{k}}\otimes\bigotimes_{k=1}^{m}a_{j_{k}}\right)=\bigotimes_{k=1}^{m}\left(\sum_{r\in J_{k}}\nu_{k,rj_{k}}a_{r}\right)$.
\end{example}
\smallskip{}

\begin{example}
\emph{Transposable elements with sex-dependent transposition rate}.

Since the Barbara McClintock's pioneering works, we know that a more
or less important part of the genome (45\% in humans) consists of
repeated sequences that can move along chromosomes. These sequences
called transposable elements or transposons (also known as mobile
genetic elements or jumping genes). The transposon displacement is
called a transposition. According to the transposon there are two
transposition modes: conservative or replicative, some transposons
use both. Transposition is conservative when the transposon moves
from one site to another without being replicated. It is replicative
when one or more copies of the transposon is transferred to another
site, in this case there is an increase of the number of copies of
the transposon inside the genome.

We consider a bisexual population composed of diploid individuals
with chromosomic number $2n\geq4$. Chromosomes are numbered from
$1$ to $n$, where $n$ is the number reserved for the $X$ and $Y$
sex chromosomes. We study the copies number of a transposon into the
genome of individuals. The notation $c_{i,j}$ is the chromosome number
$i$ carrying $j$ copies of the transposon and we assume that each
chromosome can carry at most $N$ copies.

We note $\tau_{k,j}$ (resp. $\widetilde{\tau}_{k,j}$) the insertion
(in this case $k>0$) or deletion (in this case $k<0$) probability
in a female (resp. male) of $k$ copies on a chromosome carrying $j$
copies. Thus we have $0\leq k+j\leq N$ and we put $\tau_{k,j}=\widetilde{\tau}_{k,j}=0$
as soon as $k+j<0$ or $k+j>N$, all this allows to take $-N\leq k\leq N$
in the definitions of $\tau_{k,j}$ and $\widetilde{\tau}_{k,j}$. 

Let $c_{1,i_{1}},\text{\dots},c_{n,i_{n}}$ the chromosomes transmitted
by a parent, we note $c_{i_{1},\text{\dots},i_{n}}=c_{1,i_{1}}\otimes\text{\dots}\otimes c_{n,i_{n}}$
and the genome of an individual is represented by $c_{i_{1},\text{\dots},i_{n}}\otimes c_{j_{1},\text{\dots},j_{n}}$,
where $c_{i_{1},\text{\dots},i_{n}}$ (resp. $c_{j_{1},\text{\dots},j_{n}}$)
are chromosomes inherited from the mother (resp. the father). Let
$c_{i_{1},\text{\dots},i_{n}}\otimes c_{j_{1},\text{\dots},j_{n}}$
the genome of an individual, after transposition, if the individual
is a female then the egg distribution is
\[
\frac{1}{2^{n}}\bigotimes_{p=1}^{n}\left(\sum_{k=-N}^{N}\tau_{k,i_{p}}c_{p,i_{p}+k}+\sum_{l=-N}^{N}\tau_{l,j_{p}}c_{p,j_{p}+l}\right),
\]
if the individual is a male, the distribution of spermatozoa carrying
$X$gonosome is 

\[
\frac{1}{2^{n-1}}\bigotimes_{p=1}^{n-1}\left(\sum_{k=-N}^{N}\left(\widetilde{\tau}_{k,i_{p}}c_{p,i_{p}+k}+\widetilde{\tau}_{k,j_{p}}c_{p,j_{p}+k}\right)\right)\otimes\sum_{k=-N}^{N}\widetilde{\tau}_{k,i_{n}}c_{n,i_{n}+k}
\]
 and that of the sperm cells carrying the $Y$ gonosome is 
\[
\frac{1}{2^{n-1}}\bigotimes_{p=1}^{n-1}\left(\sum_{k=-N}^{N}\left(\widetilde{\tau}_{k,i_{p}}c_{p,i_{p}+k}+\widetilde{\tau}_{k,j_{p}}c_{p,j_{p}+k}\right)\right)\otimes\sum_{k=-N}^{N}\widetilde{\tau}_{k,i_{n}}c_{n,j_{n}+k}.
\]
The result after crossing a female $c_{i_{1},\text{\dots},i_{n}}\otimes c_{j_{1},\text{\dots},j_{n}}$
with a male $c_{r_{1},\text{\dots},r_{n}}\otimes c_{s_{1},\text{\dots},s_{n}}$,
the distribution of girls is given by
\begin{eqnarray*}
\frac{1}{2^{2n-1}}\bigotimes_{p=1}^{n}\left(\sum_{k=-N}^{N}\left(\tau_{k,i_{p}}c_{p,i_{p}+k}+\tau_{k,j_{p}}c_{p,j_{p}+k}\right)\right)\otimes\qquad\qquad\qquad\quad\\
\bigotimes_{p=1}^{n-1}\left(\sum_{k=-N}^{N}\left(\widetilde{\tau}_{k,r_{p}}c_{p,r_{p}+k}+\widetilde{\tau}_{k,s_{p}}c_{p,s_{p}+k}\right)\right)\otimes\sum_{k=-N}^{N}\widetilde{\tau}_{k,r_{n}}c_{n,r_{n}+k},
\end{eqnarray*}
and that of the boys is 
\begin{eqnarray*}
\frac{1}{2^{2n-1}}\bigotimes_{p=1}^{n}\left(\sum_{k=-N}^{N}\left(\tau_{k,i_{p}}c_{p,i_{p}+k}+\tau_{k,j_{p}}c_{p,j_{p}+k}\right)\right)\otimes\qquad\qquad\qquad\quad\\
\bigotimes_{p=1}^{n-1}\left(\sum_{k=-N}^{N}\left(\widetilde{\tau}_{k,r_{p}}c_{p,r_{p}+k}+\widetilde{\tau}_{k,s_{p}}c_{p,s_{p}+k}\right)\right)\otimes\sum_{k=-N}^{N}\widetilde{\tau}_{k,s_{n}}c_{n,s_{n}+k}.
\end{eqnarray*}
Algebraically, let $A$ be the $\mathbb{R}$-vector space with basis
$\left(c_{i_{1},\text{\dots},i_{n}}\right)_{0\text{\ensuremath{\le}}i_{1},\text{\dots},i_{n}\text{\ensuremath{\le}}N}$
equipped with the algebra structure: 
\[
c_{i_{1},\text{\dots},i_{n}}c_{j_{1},\text{\dots},j_{n}}=\frac{1}{2^{n}}\bigotimes_{p=1}^{n}\left(\sum_{k=-N}^{N}\left(\tau_{k,i_{p}}c_{p,i_{p}+k}+\tau_{k,j_{p}}c_{p,j_{p}+k}\right)\right),
\]
and $\widetilde{A}$ the $\mathbb{R}$-vector space with basis $\left(\widetilde{c}_{i_{1},\text{\dots},i_{n}}\right)_{0\text{\ensuremath{\le}}i_{1},\text{\dots},i_{n}\text{\ensuremath{\le}}N}$
equipped with the algebra structure: 
\[
\widetilde{c}_{i_{1},\text{\dots},i_{n}}\widetilde{c}_{j_{1},\text{\dots},j_{n}}=\frac{1}{2^{n}}\bigotimes_{p=1}^{n}\left(\sum_{k=-N}^{N}\left(\tau_{k,i_{p}}\widetilde{c}_{p,i_{p}+k}+\tau_{k,j_{p}}\widetilde{c}_{p,j_{p}+k}\right)\right).
\]
Algebras $A$ and $\widetilde{A}$ are weighted by $\omega\left(c_{i_{1},\text{\dots},i_{n}}\right)=\widetilde{\omega}\left(\widetilde{c}_{i_{1},\text{\dots},i_{n}}\right)=1$.
By defining the maps $\varphi:A\otimes\widetilde{A}\rightarrow A$,
\[
\varphi\left(c_{i_{1},\text{\dots},i_{n}}\otimes\widetilde{c}_{j_{1},\text{\dots},j_{n}}\right)=\bigotimes_{p=1}^{n-1}\left(\sum_{k=-N}^{N}\left(\widetilde{\tau}_{k,i_{p}}c_{p,i_{p}+k}+\widetilde{\tau}_{k,i_{p}}c_{p,i_{p}+k}\right)\right)\otimes\sum_{k=-N}^{N}\widetilde{\tau}_{k,j_{n}}c_{n,j_{n}+k}
\]
 and $\widetilde{\varphi}:A\otimes\widetilde{A}\rightarrow\widetilde{A}$
by 
\[
\widetilde{\varphi}\left(c_{i_{1},\text{\dots},i_{n}}\otimes\widetilde{c}_{j_{1},\text{\dots},j_{n}}\right)=\bigotimes_{p=1}^{n-1}\left(\sum_{k=-N}^{N}\left(\widetilde{\tau}_{k,i_{p}}\widetilde{c}_{p,i_{p}+k}+\widetilde{\tau}_{k,i_{p}}\widetilde{c}_{p,i_{p}+k}\right)\right)\otimes\sum_{k=-N}^{N}\widetilde{\tau}_{k,j_{n}}\widetilde{c}_{n,j_{n}+k}
\]
we obtain an algebra structure that describes the situation exposed
above, it is gonosomal by proposition \ref{pro:AGetAlgPond}.
\end{example}
\smallskip{}

\subsection{Three constructions from two algebra laws.}
\begin{prop}
\label{prop:AG=0000261AlgPond2Lois} Let $A$ be a finite-dimensional
$K$-vector space equipped with two algebra laws $\circ$ and $\bullet$
to which $A$ is not necessarily commutative and such that $\left(A,\circ\right)$
and $\left(A,\bullet\right)$ have the same weight function $\omega$.
Then for all $\theta\in K$ and for all linear maps $\varphi,\varphi':A\rightarrow A$
verifying $\omega\circ\varphi=\omega\circ\varphi'=\omega$, the $K$-vector
space $A\times A$ with the multiplication
\begin{eqnarray*}
\left(x,y\right)\left(x',y'\right) & = & \left(\theta\varphi\left(x\circ y'+x'\circ y\right),\left(1-\theta\right)\varphi'\left(x\bullet y'+x'\bullet y\right)\right),
\end{eqnarray*}
is a gonosomal algebra.\end{prop}
\begin{proof}
It is clear that the multiplicative law defined above is commutative.
According to the hypotheses, the algebra $\left(A,\circ\right)$ admits
a basis $\left(a_{1},\ldots,a_{n}\right)$ and $\left(A,\bullet\right)$
admits a basis $\left(a'_{1},\ldots,a'_{n}\right)$ such that $\omega\left(a_{i}\right)=\omega\left(a'_{i}\right)=1$.
If we put $a_{i}\circ a'_{j}=\sum_{k=1}^{n}\gamma_{ijk}a_{k}$ and
$a_{i}\bullet a'_{j}=\sum_{k=1}^{n}\gamma'_{ijk}a'_{k}$ we have $\sum_{k=1}^{n}\gamma_{ijk}=\sum_{k=1}^{n}\gamma'_{ijk}=1$
for all $1\leq i,j\leq n$. Let $\varphi\left(a_{k}\right)=\sum_{p=1}^{n}\alpha_{pk}a_{p}$
and $\varphi'\left(a'_{k}\right)=\sum_{p=1}^{n}\alpha'_{pk}a'_{p}$,
the assumption $\omega\circ\varphi=\omega\circ\varphi'=\omega$ implies
$\sum_{p=1}^{n}\alpha_{pk}=\sum_{p=1}^{n}\alpha'_{pk}=1$ for all
$1\leq k\leq n$. For each $1\leq i\leq n$, we put $e_{i}=\left(a_{i},0\right)$
and $\widetilde{e}_{i}=\left(0,a'_{i}\right)$, it follows immediately
from the definition of the algebra law defined on $A\times A$ that
$e_{i}e_{j}=\widetilde{e}_{i}\widetilde{e}_{j}=\left(0,0\right)$
and
\begin{eqnarray*}
e_{i}\widetilde{e}_{j} & = & \left(\theta\sum_{k=1}^{n}\gamma_{ijk}\varphi\left(a_{k}\right),\left(1-\theta\right)\sum_{k=1}^{n}\gamma'_{ijk}\varphi'\left(a_{k}\right)\right)\\
 & = & \sum_{p=1}^{n}\sum_{k=1}^{n}\theta\gamma_{ijk}\alpha_{pk}e_{k}+\sum_{p=1}^{n}\sum_{k=1}^{n}\left(1-\theta\right)\gamma'_{ijk}\alpha'_{pk}\widetilde{e}_{k},
\end{eqnarray*}
 with $\sum_{p=1}^{n}\sum_{k=1}^{n}\theta\gamma_{ijk}\alpha_{pk}+\sum_{p=1}^{n}\sum_{k=1}^{n}\left(1-\theta\right)\gamma'_{ijk}\alpha'_{pk}=1$.\end{proof}
\begin{example}
\emph{Heredity of an autosomal gene with sex-dependent segregation
and mutation}.
\end{example}
We consider the alleles $a_{1},\ldots,a_{n}$ an autosomal gene in
a diploid bisexual population. Genotypes $\left(a_{i},a_{j}\right)$
are ordered, i.e. in $\left(a_{i},a_{j}\right)$ allele $a_{i}$ is
transmitted by the mother and $a_{j}$ is received from the father.
It is assumed that the segregation coefficients and the mutation rates
during meiosis are sex-dependent, we note $\eta_{ijk}$ (resp. $\widetilde{\eta}_{ijk}$)
the segregation coefficient of allele $a_{k}$ in a female (resp.
male) $\left(a_{i},a_{j}\right)$ and $\mu_{pk}$ (resp. $\widetilde{\mu}_{pk}$)
the mutation rate of allele $a_{k}$ to $a_{p}$ in females (resp.
males). Finally it is assumed that each generation the female rate
in the population remains constant equal to $\theta$. Then after
crossing a female $\left(a_{i},a_{j}\right)$ with a male $\left(a_{p},a_{q}\right)$,
the proportion of females (resp. males) $\left(a_{k},a_{r}\right)$
in the offspring is $\theta\sum_{p,q=1}^{n}\mu_{pi}\mu_{qj}\eta_{ijk}\eta_{pqr}$
(resp. $\left(1-\theta\right)\sum_{p,q=1}^{n}\widetilde{\mu}_{pi}\widetilde{\mu}_{qj}\widetilde{\eta}_{ijk}\widetilde{\eta}_{pqr}$).

Algebraically, we define on the vector space $V$ with basis $\left(a_{1},\ldots,a_{n}\right)$
the algebra laws: $a_{i}\circ a_{j}=\sum_{k=1}^{n}\eta_{ijk}a_{k}$
and $a_{i}\bullet a_{j}=\sum_{k=1}^{n}\widetilde{\eta}_{ijk}a_{k}$
where $\sum_{k=1}^{n}\eta_{ijk}=\sum_{k=1}^{n}\widetilde{\eta}_{ijk}=1$,
algebras $\left(V,\circ\right)$ and $\left(V,\bullet\right)$ are
weighted by $\omega\left(a_{i}\right)=1$. Then by applying the proposition
\ref{prop:AG=0000261AlgPond2Lois} with $A=V\otimes V$ the non commutative
duplicate of $V$ equipped with laws $\left(x\otimes y\right)\circ\left(x'\otimes y'\right)=\left(x\circ y\right)\otimes\left(x'\circ y'\right)$
and $\left(x\otimes y\right)\bullet\left(x'\otimes y'\right)=\left(x\bullet y\right)\otimes\left(x'\bullet y'\right)$,
the linear maps $\varphi,\varphi':A\rightarrow A$ defined by $\varphi\left(a_{i}\otimes a_{j}\right)=\sum_{p,q=1}^{n}\mu_{pi}\mu_{qj}a_{p}\otimes a_{q}$,
$\varphi'\left(a_{i}\otimes a_{j}\right)=\sum_{p,q=1}^{n}\widetilde{\mu}_{pi}\widetilde{\mu}_{qj}a_{p}\otimes a_{q}$,
we obtain the frequency distributions of genetic types after crossing
a female $e_{ij}=\left(a_{i}\otimes a_{j},0\right)$ with a male $\widetilde{e}_{ij}=\left(0,a_{i}\otimes a_{j}\right)$.
\begin{example}
\emph{Reproduction in ciliates}.
\end{example}
Ciliates (\emph{Ciliophora}) are unicellular organisms. Ciliates can
reproduce in two ways: by fission (asexual reproduction) or by conjugation
(called sexual reproduction). Ciliates have two nuclei: a large nucleus
(macronucleus) polyploid which is involved in the control of metabolism
and a smaller nucleus (micronucleus) diploid which is involved in
reproduction.

Conjugation is a cross-fertilization process, it begins when two ciliates
come together side by side and form a cytoplasmic bridge between them,
this cytoplasmic bridge ensures the exchange of genetic material.
The macronuclei degenerate while each micronucleus undergoes two meiosis
resulting in four haploid nuclei which three are destroyed, the remaining
nucleus divides to produce two haploid gametic nuclei: a resident
and a mobile. The mobile nuclei are exchanged between the two conjugants,
then the conjugants separate, the gametic nuclei fuse in a zygotic
nucleus which divides several times, among the nuclei obtained one
becomes micronucleus and the others form a macronucleus. Finally,
after conjugation we have two genetically identical individuals.

Given a diploid ciliate species with $2n$ chromosomes in the micronucleus.
For each $1\leq k\leq n$, we note $a_{1}^{k},\ldots,a_{r_{k}}^{k}$
the genetic types that can be found on the chromosome $k$ and $A_{k}$
the vector space with basis $\left(a_{1}^{k},\ldots,a_{r_{k}}^{k}\right)$
equipped with the algebra structure $a_{i}^{k}a_{j}^{k}=\frac{1}{2}a_{i}^{k}+\frac{1}{2}a_{j}^{k}$.
We define on the space $\otimes_{k=1}^{n}A_{k}$ the algebra structure:
$\left(\otimes_{k=1}^{n}a_{i_{k}}^{k}\right)*\left(\otimes_{k=1}^{n}a_{j_{k}}^{k}\right)=\otimes_{k=1}^{n}\left(a_{i_{k}}^{k}a_{j_{k}}^{k}\right)$
which gives the distribution of gametic nuclei produced by ciliates
whose micronucleus genotype is $\left(a_{i_{1}}^{1},\ldots,a_{i_{n}}^{n}/a_{j_{1}}^{1},\ldots,a_{j_{n}}^{n}\right)$.
Then by applying the proposition \ref{prop:AG=0000261AlgPond2Lois}
with $A=\left(\otimes_{k=1}^{n}A_{k}\right)\otimes\left(\otimes_{k=1}^{n}A_{k}\right)$,
taking as laws $\circ$ and $\bullet$ the law defined by $\left(x\otimes y\right)\left(x'\otimes y\right)=\left(x*y\right)\otimes\left(x'*y'\right)$.
If there is no mutation, we take the maps $\varphi=\psi=Id$ and $\theta=\frac{1}{2}$,
we obtain the genotypes distribution after conjugation of two ciliates.
\begin{prop}
\label{prop:AG=0000262lois_nonpond} Let $A$ be a $K$-vector space
equipped with two algebra laws $\circ$ and $\bullet$ to which $A$
is not necessarily commutative or baric. The $K$-space $A\otimes A$
is equipped with the multiplication: 
\[
\left(a\otimes b\right)*\left(c\otimes d\right)=\left(a\circ b\right)\otimes\left(c\bullet d\right).
\]
If there is a subalgebra $G$ of $A\otimes A$ such that $G$ is finite-dimensional,
weighted by a map $\omega$, then for all $\theta\in K$ and for all
linear maps $\varphi,\psi:G\rightarrow G$ such that $\omega\circ\varphi=\omega\circ\psi=\omega$,
the $K$-space $G\times G$ equipped with the law: 
\[
\left(x,y\right)\left(x',y'\right)=\left(\theta\varphi\left(x*y'+x'*y\right),\left(1-\theta\right)\psi\left(x*y'+x'*y\right)\right)
\]
is a gonosomal algebra.\end{prop}
\begin{proof}
The algebra $G$ admits a basis $\left(a_{1},\ldots,a_{n}\right)$
such that $\omega\left(a_{i}\right)=1$ for all $1\leq i\leq n$.
It follows that if $a_{i}*a_{j}=\sum_{k=1}^{n}\gamma_{ijk}a_{k}$
then $\sum_{k=1}^{n}\gamma_{ijk}=1$ for $1\leq i,j\leq n$. Next
if $\varphi\left(a_{i}\right)=\sum_{k=1}^{n}\alpha_{ki}a_{k}$ and
$\psi\left(a_{i}\right)=\sum_{k=1}^{n}\beta_{ki}a_{k}$, from $\omega\circ\varphi=\omega\circ\psi=\omega$
it comes $\sum_{k=1}^{n}\alpha_{ki}=\sum_{k=1}^{n}\beta_{ki}=1$ for
all $1\leq i\leq n$. It is clear that the law defined on $G\times G$
is commutative. If for all $1\leq i\leq n$ we put $e_{i}=\left(a_{i},0\right)$
and $\widetilde{e}_{i}=\left(0,a_{i}\right)$, then we have $e_{i}e_{j}=0$,
$\widetilde{e}_{p}\widetilde{e}_{q}=0$ and
\[
e_{i}\widetilde{e}_{p}=\theta\sum_{k,r=1}^{n}\gamma_{ipk}\alpha_{rk}e_{r}+\left(1-\theta\right)\sum_{k,r=1}^{n}\gamma_{ijk}\beta_{kr}\widetilde{e}_{r},
\]
with $\theta\sum_{k,r=1}^{n}\gamma_{ipk}\alpha_{rk}+\left(1-\theta\right)\sum_{k,r=1}^{n}\gamma_{ijk}\beta_{kr}=1$.
\end{proof}
\smallskip{}

\begin{example}
\emph{Genomic imprinting (or parental imprinting)}
\end{example}
In many diploid placental mammals, we observe a functional asymmetry
for some autosomal genes according to their paternal or maternal origin,
it results in the offspring by the expression of only one allele on
both. This phenomenon is called \emph{genomic (or parental) imprinting}.
The genomic imprinting is submitted to the cycle: deletion, installation,
maintenance. In every generation in each individual at the time of
gamete formation, all parental imprints are erased in the germ cells
and when these cells become mature gametes the genes subject to imprinting
are inactivated or not according to the sex of the individual, next
the imprints are transmitted by fertilization to the next generation
where they are transmitted through cell divisions throughout the life.

For example, the mutation \emph{brachyury} ($T$) of the mouse is
known since 1927, it is a dominant mutation that results in a shortening
of the tail. One of its alleles, called hairpin tail (denoted $T^{hp}$)
has a strange inheritance: the cross $\left(\text{\Female},+/+\right)\times\left(\text{\Male},T^{hp}/+\right)$,
sign $+$ denoting the normal allele, gives a descent made up 50\%
of short-tailed mice and 50\% of normal tail, while the symmetric
cross $\left(\text{\Female},T^{hp}/+\right)\times\left(\text{\text{\Male}},+/+\right)$
gives only normal tail mice. This observation which is contradictory
to the Mendel laws, is explained by maternal imprint which inactivates
allele $T^{hp}$ and thereby silencing it, which means that the embryos
$T^{hp}/+$ die in utero. It follows that we cannot observe homozygote
$T^{hp}/T^{hp}$, so the cross $\left(T^{hp}/+\right)\times\left(T^{hp}/+\right)$
produces only $T^{hp}/+$ and $+/+$ descendants.

Algebraically, we establish coding $e_{1}\leftrightarrow+$, $e_{2}\leftrightarrow T^{hp}$,
we define on the space $A$ with basis $\left(e_{1},e_{2}\right)$
the laws $\circ$ and $\bullet$ :

\hspace{2cm}%
\begin{tabular}{lll}
$e_{1}\circ e_{1}=e_{1}$, & $e_{2}\circ e_{2}=0$,  & $e_{1}\circ e_{2}=e_{2}\circ e_{1}=e_{1}$.\tabularnewline
$e_{1}\bullet e_{1}=e_{1}$, & $e_{2}\bullet e_{2}=0$,  & $e_{1}\bullet e_{2}=e_{2}\bullet e_{1}=\frac{1}{2}e_{1}+\frac{1}{2}e_{2}$.\tabularnewline
\end{tabular}

If we note $e_{ij}=e_{i}\otimes e_{j}$, by applying the proposition
\ref{prop:AG=0000262lois_nonpond} we have 
\begin{eqnarray*}
e_{11}e_{ij}=e_{12}e_{ij}=e_{21}e_{ij} & = & \begin{cases}
e_{11} & \mbox{si }\left(i,j\right)=\left(1,1\right)\\
\frac{1}{2}e_{11}+\frac{1}{2}e_{12} & \mbox{si }\left(i,j\right)=\left(1,2\right),\left(2,1\right)\\
0 & \mbox{si }\left(i,j\right)=\left(2,2\right)
\end{cases}\\
e_{22}e_{ij} & = & 0
\end{eqnarray*}
and taking for $G$ the space with basis $\left(e_{11},e_{12}\right)$
weighted by $\omega\left(e_{11}\right)=\omega\left(e_{12}\right)=1$,
the maps $\varphi=\psi=Id$ and $\theta$ the female proportion by
generation, we obtain an algebraic model for the transmission of \emph{brachyury
}mutation.

\smallskip{}

\begin{prop}
\label{prop:AG=000026EV2lois_sanspond} Let $A$ be a $K$-vector
space, $A_{1},A_{2}\subset A$, $A_{1}\cap A_{2}=\left\{ 0\right\} $
two finite-dimensional subspaces with respective bases $\left(a_{1,i}\right)_{1\leq i\leq n_{1}}$
and $\left(a_{2,i}\right)_{1\leq i\leq n_{2}}$. If it exists on $A_{1}\oplus A_{2}$
two algebra laws $\circ$ and $\bullet$ verifying $A_{1}\circ A_{1}\cup A_{1}\circ A_{2}\subset A_{1}$,
$A_{1}\bullet A_{2}\cup A_{2}\bullet A_{2}\subset A_{2}$ and if the
linear map $\eta:A_{1}+A_{2}\rightarrow K$ defined by $\eta\left(a_{1,i}\right)=\eta\left(a_{2,j}\right)=1$
verify $\eta\left(a_{1,i}\circ a_{2,j}\right)=1$, $\eta\left(a_{1,i}\bullet a_{2,j}\right)=1$,
then for all $\theta\in K$ and for all linear maps $\varphi_{1}:A_{1}\rightarrow A_{1}$
and $\varphi_{2}:A_{2}\rightarrow A_{2}$ such that $\eta\circ\varphi_{1}=\eta\circ\varphi_{2}=\eta$,
the space $A_{1}\times A_{2}$ with multiplication
\[
\left(x_{1},x_{2}\right)\left(y_{1},y_{2}\right)=\left(\theta\varphi_{1}\left(x_{1}\circ y_{2}+y_{1}\circ x_{2}\right),\left(1-\theta\right)\varphi_{2}\left(x_{1}\bullet y_{2}+y_{1}\bullet x_{2}\right)\right)
\]
is a gonosomal algebra.\end{prop}
\begin{proof}
This multiplication defined on $A_{1}\times A_{2}$ is commutative.
If we put $a_{1,i}\circ a_{2,j}=\sum_{k=1}^{n_{1}}\gamma_{ijk}a_{1,k}$
and $a_{1,i}\bullet a_{2,j}=\sum_{k=1}^{n_{2}}\gamma'_{ijk}a_{2,k}$,
from $\eta\left(a_{1,i}\circ a_{2,j}\right)=1$, $\eta\left(a_{1,i}\bullet a_{2,j}\right)=1$
it comes $\sum_{k=1}^{n_{1}}\gamma_{ijk}=\sum_{k=1}^{n_{2}}\gamma'_{ijk}=1$.
If $\left(\alpha_{ij}\right)_{1\leq i,j\leq n_{1}}$ and $\bigl(\alpha'_{ij}\bigr)_{1\leq i,j\leq n_{2}}$
are respectively matrices of the maps $\varphi_{1}$ and $\varphi_{2}$
in bases $\left(a_{1,i}\right)_{1\leq i\leq n_{1}}$ and $\left(a_{2,i}\right)_{1\leq i\leq n_{2}}$,
from $\eta\circ\varphi_{1}=\eta\circ\varphi_{2}=\eta$ it comes $\sum_{i=1}^{n_{1}}\alpha_{ij}=1$
and $\sum_{i=1}^{n_{2}}\alpha'_{ij}=1$. Then if we put $e_{i}=\left(a_{1,i},0\right)$
and $\widetilde{e}_{p}=\left(0,a_{2,p}\right)$ we have $e_{i}e_{j}=0$,
$\widetilde{e}_{p}\widetilde{e}_{q}=0$ and
\begin{eqnarray*}
e_{i}\widetilde{e}_{p} & = & \left(\theta\sum_{k=1}^{n_{1}}\gamma_{ipk}\varphi_{1}\left(a_{1,k}\right),\left(1-\theta\right)\sum_{k=1}^{n_{2}}\gamma'_{ipk}\varphi_{2}\left(a_{2,k}\right)\right)\\
 & = & \theta\sum_{k,l=1}^{n_{1}}\gamma_{ipk}\alpha_{lk}e_{l}+\left(1-\theta\right)\sum_{k,l=1}^{n_{2}}\gamma'_{ipk}\alpha'_{lk}\widetilde{e}_{l}
\end{eqnarray*}
with $\theta\sum_{k,l=1}^{n_{1}}\gamma_{ipk}\alpha_{lk}+\left(1-\theta\right)\sum_{k,l=1}^{n_{2}}\gamma'_{ipk}\alpha'_{lk}=1$.
\end{proof}
\smallskip{}

\begin{example}
\emph{X-inactivation (or lyonization)}.
\end{example}
In most placental or marsupial mammals, one of the $X$ chromosome
in the female genome is inactive: genes carried by this chromosome
are not expressed throughout the lifetime. The $X$-inactivation occurs
upon implantation of the egg, it is random in placental mammals and
in marsupials it is always the $X$ chromosome inherited from the
father who is inactived.

We use the proposition \ref{prop:AG=000026EV2lois_sanspond} to give
an algebraic model of the $X$-inactivation. We consider the alleles
$a_{1},\ldots,a_{n}$ of a gonosomal gene, we note $a_{1},\ldots,a_{n}$
(resp. $a_{1}^{*},\ldots,a_{n}^{*}$) when these alleles are active
(resp. silencer). Let $V$ be the vector space with basis $\left(a_{i},a_{i}^{*}\right)_{1\leq i\leq n}$
and $A=V+V\otimes V$. The space $A_{1}$ spanned by $\left\{ a_{i}^{*}\otimes a_{j},a_{i}\otimes a_{j}^{*};1\leq i,j\leq n\right\} $
represents ordered female genotypes, the space $A_{2}=V$ gives male
genotypes. If the $X$ chromosome inactivation rate of maternal origin
is noted by $\tau$, we define on $A_{1}\oplus A_{2}$ the laws $\circ$
and $\bullet$ : 
\begin{eqnarray*}
a_{i}^{*}\otimes a_{j}\circ a_{k}=a_{i}\otimes a_{j}^{*}\circ a_{k} & = & \frac{\tau}{2}\left(a_{i}^{*}\otimes a_{k}+a_{j}^{*}\otimes a_{k}\right)+\frac{1-\tau}{2}\left(a_{i}\otimes a_{k}^{*}+a_{j}\otimes a_{k}^{*}\right)\\
a_{i}\bullet a_{j} & = & \frac{1}{2}\left(a_{i}+a_{j}\right).
\end{eqnarray*}

If $\theta$ means the rate of females in the population, and if we
take $\varphi_{i}=id_{A_{i}}$, $i=1,2$, then the multiplication
defined in the proposition \ref{prop:AG=000026EV2lois_sanspond} gives
the genotype distribution in the absence of mutation of a cross between
a female genotype $\left(a_{i}^{*}\otimes a_{j},0\right)$ or $\left(a_{i}\otimes a_{j}^{*},0\right)$
with a male genotype $\left(0,a_{k}\right)$. In the presence of mutation,
we define $\varphi_{1}\left(a_{i}^{*}\otimes a_{j}\right)=\sum_{p=1}^{n}\mu_{kj}a_{i}^{*}\otimes a_{k}$,
$\varphi_{1}\left(a_{i}\otimes a_{j}^{*}\right)=\sum_{p=1}^{n}\mu_{ki}a_{i}\otimes a_{k}^{*}$
and $\varphi_{2}\left(a_{i}\right)=\sum_{p=1}^{n}\widetilde{\mu}_{ki}a_{k}$
where $\mu_{ki}$ (resp. $\widetilde{\mu}_{ki}$) is the mutation
rate of the allele $a_{i}$ to the allele $a_{k}$ in females (resp.
males).
\begin{example}
\emph{Sex determination by elimination of sex chromosomes}.

In \emph{Sciaridae} sex is determined by the gonosome $X$: females
are $XX$ and males $X0$ with the peculiarity that \emph{Sciaridae}
males are obtained by elimination of all chromosomes coming from the
father. We note respectively $A^{m}$ and $A^{p}$ the set of autosomes
coming from the mother and father, $X^{m}$ and $X^{p}$ a gonosome
$X$ transmitted by the mother and father. In Sciaridae after fertilization
zygotes are of genotype $A^{m}A^{p}X^{m}X^{p}X^{p}$, during the formation
of the somatic lineage, one gonosome $X^{p}$ is eliminated from the
cells of the somatic female lineages while both $X^{p}$ are eliminated
in the somatic male lineages. In male germline during spermatogenesis
cells lose chromosomes $A^{p}$ and $X^{p}$, the $X^{m}$ gonosome
is replicated, finally we obtain a $A^{m}X^{m}X^{m}$ type spermatozoa.
In females, oogenesis proceeds normally and it leads to $AX$ type
ovules.

We must add that according to the composition of female progeny two
types of Sciaridae are distinguished: digenic or monogenic. In digenic
Sciaridae the progeny of a female consists of males and females with
a sex-ratio different from 1. In monogenic Sciaridae, descendants
are all of the same sex, the females are called androgenic when they
produce only males and gynogenic when they produce only females. There
are also Sciaridae which are monogenic and digenic.

\smallskip{}

Sex determination in digenic Sciaridae can be represented algebraically
as follows. Let $V$, $G$, $O$ spaces with respective bases $\left(a_{1},\ldots,a_{n}\right)$,
$\left(g_{1},\ldots,g_{m}\right)$ and $\left(o\right)$. We consider
the space $A=V\otimes V\otimes G\otimes\left(G+O\right)$. An element
$a_{i}\otimes a_{j}\otimes g_{k}\otimes g$ of the basis of $A$,
or $\left(a_{i}a_{j}g_{k}g\right)$ in an abbreviated form, is the
genotype of an adult, where $a_{i}$ (resp. $a_{j}$) represents autosomes
coming from the mother (resp. father), $g_{k}$ a gonosome $X$ transmitted
by the mother and $g$ indicate the sex of the individual with $g\in\left\{ g_{1},\ldots,g_{m}\right\} $
if female and in this case $g$ was transmitted by the father, or
$g=o$ if male. We apply the proposition \ref{prop:AG=000026EV2lois_sanspond}
by taking $A_{1}=V\otimes V\otimes G\otimes G$, $A_{2}=V\otimes V\otimes G\otimes O$
equipped with multiplications:
\begin{eqnarray*}
\left(a_{i}a_{j}g_{k}g\right)\circ\left(a{}_{p}a{}_{q}g{}_{r}g'\right) & = & \begin{cases}
\frac{1}{4}\left(a_{i}+a_{j}\right)\otimes a{}_{p}\otimes\left(g_{k}+g\right)\otimes g{}_{r} & \mbox{if }\left(g,g'\right)\in G\times O,\\
0 & \mbox{otherwise,}
\end{cases}\\
\left(a_{i}a_{j}g_{k}g\right)\bullet\left(a{}_{p}a{}_{q}g{}_{r}g'\right) & = & \begin{cases}
\frac{1}{4}\left(a_{i}+a_{j}\right)\otimes a{}_{p}\otimes\left(g_{k}+g\right)\otimes o & \mbox{if }\left(g,g'\right)\in G\times O,\\
0 & \mbox{otherwise,}
\end{cases}
\end{eqnarray*}
with $\eta\left(a_{i}a_{j}g_{k}g_{l}\right)=\eta\left(a_{i}a_{j}g_{k}o\right)=1$,
$\varphi=\psi=Id$ and $\theta=\frac{1}{1+\sigma}$ where $\sigma$
is the sex-ratio of the population.

\smallskip{}

The Sciaridae monogeny depends on a particular gonosome $X$ noted
$X^{*}$: $X^{*}X$ females are gynogenic while $XX$ females are
androgenic. Algebraically, let $V$, $G^{*}$, $G$, $O$ spaces with
respective bases $\left(a_{1},\ldots,a_{n}\right)$, $\left(g_{1}^{*},\ldots,g_{p}^{*}\right)$,
$\left(g_{1},\ldots,g_{m}\right)$ and $\left(o\right)$, where $a_{k}$
represents an autosomal type, $g_{k}^{*}$ a gonosome determining
gynogeny, $g_{k}$ a $X$ gonosome type and $o$ is associated with
the male sex. Applying the proposition \ref{prop:AG=000026EV2lois_sanspond}
with the spaces $A=V\otimes V\otimes\left(G^{*}+G\right)\otimes\left(G+O\right)$,
$A_{1}=V\otimes V\otimes\left(G^{*}+G\right)\otimes G$, $A_{2}=V\otimes V\otimes G\otimes O$,
we define the laws $\circ$ and $\bullet$ by :
\begin{eqnarray*}
\left(a_{i}a_{j}g_{k}^{*}g\right)\circ\left(a{}_{p}a{}_{q}g{}_{r}g'\right) & = & \frac{1}{4}\left(a_{i}+a_{j}\right)\otimes a{}_{p}\otimes\left(g_{k}^{*}+g\right)\otimes g{}_{r}\mbox{ if }\left(g,g'\right)\in G\times O,\\
\left(a_{i}a_{j}g_{k}g_{l}\right)\circ\left(a{}_{p}a{}_{q}g{}_{r}g'\right) & = & 0\mbox{ if }g'\in G+O,\\
\left(a_{i}a_{j}g_{k}^{*}g_{l}\right)\bullet\left(a{}_{p}a{}_{q}g{}_{r}g'\right) & = & 0\mbox{ if }g'\in G+O,\\
\left(a_{i}a_{j}g_{k}g\right)\bullet\left(a{}_{p}a{}_{q}g{}_{r}g'\right) & = & \frac{1}{4}\left(a_{i}+a_{j}\right)\otimes a{}_{p}\otimes\left(g_{k}+g\right)\otimes o\mbox{ if }\left(g,g'\right)\in G\times O,
\end{eqnarray*}
with $\eta\left(a_{i}a_{j}g_{k}g_{l}\right)=\eta\left(a_{i}a_{j}g_{k}^{*}g_{l}\right)=\eta\left(a_{i}a_{j}g_{k}o\right)=1$;
$\varphi=\psi=Id$ and $0\leq\theta\leq1$ corresponds to the proportion
of gynogeny ($\theta=1$ in the case of gynogenic \emph{Sciaridae},
$\theta=0$ for monogenic).
\end{example}
\smallskip{}

\subsection{Construction from three linear forms and three linear maps.}

\textcompwordmark{}\smallskip{}

This construction applies to the case of ordered genotypes where take
into account of maternal and paternal origin of genes.
\begin{prop}
\label{prop:AGet3FormLin} Let $A$, $A'$, $\widetilde{A}$ be finite-dimensional
$K$-vector spaces; $\omega:A\rightarrow K$, $\omega':A'\rightarrow K$,
$\widetilde{\omega}:\widetilde{A}\rightarrow K$ nonzero linear forms
and $\varphi:A\otimes A'\rightarrow A$, $\varphi':A\otimes\widetilde{A}\rightarrow A'$,
$\widetilde{\varphi}:A\otimes\widetilde{A}\rightarrow\widetilde{A}$
three linear maps such that $\omega\circ\varphi=\omega\otimes\omega'$
and $\omega'\circ\varphi'+\widetilde{\omega}\circ\widetilde{\varphi}=\omega\otimes\widetilde{\omega}$,
then the $K$-space $A\otimes A'\oplus A\otimes\widetilde{A}$ equipped
with the algebra structure:
\[
\left(x\oplus y\right)\left(x'\oplus y'\right)=\left[\varphi\left(x\right)\otimes\varphi'\left(y'\right)+\varphi\left(x'\right)\otimes\varphi'\left(y\right)\right]\oplus\left[\varphi\left(x\right)\otimes\widetilde{\varphi}\left(y'\right)+\varphi\left(x'\right)\otimes\widetilde{\varphi}\left(y\right)\right]
\]
 is a agonosomal algebra.\end{prop}
\begin{proof}
Note that for all $x\in A\otimes A'$ and $y\in A\otimes\widetilde{A}$,
identifying $x\oplus0$ to $x$ and $0\oplus y'$ to $y'$, the multiplication
$A\otimes A'\oplus A\otimes\widetilde{A}$ given in the statement
becomes: 
\[
xy'=\varphi\left(x\right)\otimes\varphi'\left(y'\right)\oplus\varphi\left(x\right)\otimes\widetilde{\varphi}\left(y'\right),\qquad\left(*\right)
\]
and furthermore we have $y'x=xy'$, $xx'=yy'=0$ for all $x,x'\in A\otimes A'$
and $y,y'\in A\otimes\widetilde{A}$. 

Linear forms $\omega$, $\omega'$, $\widetilde{\omega}$ being nonzero
and spaces $A$, $A'$, $\widetilde{A}$ being finite-dimensional,
they admit bases $\left(a_{i}\right)_{1\leq i\leq n}$, $\bigl(a'_{j}\bigr)_{1\leq j\leq m}$
and $\left(\widetilde{a}_{k}\right)_{1\leq k\leq p}$ such that $\omega\left(a_{i}\right)=1$,
$\omega\bigl(a'_{j}\bigr)=1$ and $\omega\left(\widetilde{a}_{k}\right)=1$
(cf. proof of the proposition \ref{pro:AGetDupliquee}). Let $\sigma:\left[\!\left[1,n\right]\!\right]\times\left[\!\left[1,m\right]\!\right]\rightarrow\left[\!\left[1,nm\right]\!\right]$
and $\tau:\left[\!\left[1,n\right]\!\right]\times\left[\!\left[1,p\right]\!\right]\rightarrow\left[\!\left[1,np\right]\!\right]$
bijections, we define the maps $\varphi\bigl(a_{i}\otimes a'_{j}\bigr)=\sum_{r=1}^{n}\lambda_{r,\sigma\left(i,j\right)}a_{r}$,
$\varphi'\left(a_{i}\otimes\widetilde{a}{}_{k}\right)=\sum_{s=1}^{m}\mu_{s,\tau\left(i,k\right)}a'_{s}$
and $\widetilde{\varphi}\left(a_{i}\otimes\widetilde{a}{}_{k}\right)=\sum_{t=1}^{p}\nu_{t,\tau\left(i,k\right)}\widetilde{a}{}_{t}$.
With this, according to $\left(*\right)$ we have:
\[
\left(a_{i}\otimes a'_{j}\right)\left(a_{l}\otimes\widetilde{a}{}_{k}\right)=\sum_{r=1}^{n}\sum_{s=1}^{m}\lambda_{r,\sigma\left(i,j\right)}\mu_{s,\tau\left(l,k\right)}a_{r}\otimes a'_{s}\oplus\sum_{r=1}^{n}\sum_{t=1}^{p}\lambda_{r,\sigma\left(i,j\right)}\nu_{t,\tau\left(l,k\right)}a_{r}\otimes\widetilde{a}{}_{t},\quad\left(**\right)
\]
but we have
\[
\sum_{r=1}^{n}\lambda_{r,\sigma\left(i,j\right)}=\omega\left(\varphi\bigl(a_{i}\otimes a'_{j}\bigr)\right)=\omega\otimes\omega'\bigl(a_{i}\otimes a'_{j}\bigr)=\omega\left(a_{i}\right)\omega'\left(a'_{j}\right)=1,
\]
and also
\[
\sum_{s=1}^{n}\mu_{s,\tau\left(l,j\right)}+\sum_{t=1}^{p}\nu_{t,\tau\left(l,k\right)}=\omega'\left(\varphi'\left(a_{i}\otimes\widetilde{a}{}_{k}\right)\right)+\widetilde{\omega}\left(\widetilde{\varphi}\left(a_{i}\otimes\widetilde{a}{}_{k}\right)\right)=\omega\otimes\omega'\left(a_{i}\otimes\widetilde{a}{}_{k}\right)=1,
\]
hence 
\[
\sum_{r=1}^{n}\sum_{s=1}^{m}\lambda_{r,\sigma\left(i,j\right)}\mu_{s,\tau\left(l,k\right)}+\sum_{r=1}^{n}\sum_{t=1}^{p}\lambda_{r,\sigma\left(i,j\right)}\nu_{t,\tau\left(l,k\right)}=1.
\]
 Finally, by putting $e_{\sigma\left(i,j\right)}=a_{i}\otimes a'_{j}$,
$\widetilde{e}_{\tau\left(l,k\right)}=a_{l}\otimes\widetilde{a}{}_{k}$,
$\gamma_{\sigma\left(i,j\right),\tau\left(l,k\right),\sigma\left(r,s\right)}=\lambda_{r,\sigma\left(i,j\right)}\mu_{s,\tau\left(l,k\right)}$
and $\widetilde{\gamma}_{\sigma\left(i,j\right),\tau\left(l,k\right),\tau\left(r,t\right)}=\lambda_{r,\sigma\left(i,j\right)}\nu_{t,\tau\left(l,k\right)}$,
the relation $\left(**\right)$ is written in the form: 
\[
e_{\sigma\left(i,j\right)}\widetilde{e}_{\tau\left(l,k\right)}=\sum_{\sigma\left(r,s\right)=1}^{nm}\gamma_{\sigma\left(i,j\right),\tau\left(l,k\right),\sigma\left(r,s\right)}e_{\sigma\left(r,s\right)}+\sum_{\tau\left(r,t\right)=1}^{np}\widetilde{\gamma}_{\sigma\left(i,j\right),\tau\left(l,k\right),\tau\left(r,t\right)}\widetilde{e}_{\tau\left(r,t\right)}
\]
with $\sum_{\sigma\left(r,s\right)=1}^{nm}\gamma_{\sigma\left(i,j\right),\tau\left(l,k\right),\sigma\left(r,s\right)}+\sum_{\tau\left(r,t\right)=1}^{np}\widetilde{\gamma}_{\sigma\left(i,j\right),\tau\left(l,k\right),\tau\left(r,t\right)}=1$,
which establishes that the algebra $A\otimes A'\oplus A\otimes\widetilde{A}$
is gonosomal. 
\end{proof}
We are going to apply this result to to an exceptional mode of reproduction
in the living world.
\begin{example}
\emph{Reproduction of a triploid} (\emph{Bufo baturae}).

In nature triploid individuals appear from a cross between a tetraploid
($4n$) and a diploid ($2n$) as the result of the union between a
diploid ($2n$) and a haploid ($n$) gamete which gives a triploid
zygote ($3n$). Triploids are sterile, this is due to the mechanism
of gamete formation. Indeed, during meiosis chromosomes are associated
by homologous pairs before being divided equally for each group in
different gametes and this independently for each group. As a result,
in triploids for each homologous pair, a gamete receives two chromosomes
and an other only one. Since the distribution of chromosomes is independent
of a homologous pair to another, the probability of having a haploid
or diploid gamete in a $3n$ triploid is $\left(\nicefrac{1}{2}\right)^{n-1}$,
for example for $3n=33$ this probability is $\nicefrac{1}{1024}$.
In theory there exists a possibility that two haploid or diploid gametes
unite to give an euploid zygote, but the probability of this event
is very low. The triploid sterility is used in agriculture to produce
seedless fruits (banana, watermelon, grapes \dots ). 

From the above we can understand the surprise of biologists discovering
that the toad \emph{Bufo baturae}, who lives in the desert mountains
of northern Pakistan, is triploid ($3n=33$) and reproduces. The explanation
of this particularism was found in male and female gametogenesis.
The \emph{B. baturae} genome contains two sets of chromosomes carrying
a nucleolar organizer (nucleolus-organizing region or $NOR$) and
one set without $NOR$, noted $NOR^{-}$ in the following. In males,
the cells at the origin of the sperm undergo at first elimination
of the set of $NOR^{-}$ chromosomes, so we obtain diploid spermatocytes
which, after a normal meiosis with possibilities of recombination,
give haploid sperm of $NOR^{+}$ type. In females, there are two kinds
of ovogenesis. In one, as in males, there is at first elimination
of one genome $NOR^{+}$, followed by the duplication of the remaining
chromosomes which leads to the formation of tetraploid cells, followed
by a meiosis without recombinations. In the other, there is at first
duplication of the set $NOR^{-}$ followed by a normal meiosis with
recombination between $NOR^{+}$ chromosomes. In both cases it leads
to diploid eggs of $NOR^{+}/NOR^{-}$ type whose fusion with the sperm
gives triploid zygotes containing two sets of $NOR^{+}$ and one of
$NOR^{-}$.\smallskip{}

We can apply the proposition \ref{prop:AGet3FormLin} to this situation.
Note $\left(a_{i}\right)_{1\leq i\leq n}$ (resp. $\left(\widetilde{a}_{j}\right)_{1\leq j\leq m}$)
the sets of chromosomes $NOR^{+}$ containing the chromosome $X$
(resp. $Y$) and $\left(a_{k}^{-}\right)_{1\leq k\leq p}$ the sets
$NOR^{-}$. We consider the spaces $A$ with basis $\left(a_{i}\otimes a_{k}^{-}\right)_{i,k}$,
$A'$ with basis $\left(a_{i}\right)_{i}$ and $\widetilde{A}$ with
basis $\left(\widetilde{a}_{j}\right)_{j}$, we give the linear forms
$\omega\left(a_{i}\otimes a_{k}^{-}\right)=1$, $\omega'\left(a_{i}\right)=1$
and $\widetilde{\omega}\left(\widetilde{a}_{j}\right)=1$. In this
model an element $a_{i}\otimes a_{k}^{-}\otimes a_{j}$ (resp. $a_{i}\otimes a_{k}^{-}\otimes\widetilde{a}_{j}$)
in the basis of the space $A\otimes A'$ (resp. $A\otimes\widetilde{A}$)
represents a female (resp male) karyotype in which sets $a_{i}$,
$a_{k}^{-}$ were transmitted by the mother and sets $a_{j}$ and
$\widetilde{a}_{j}$ come from father. 

Let $\theta$ the rate of oocytes losing a set $NOR^{+}$ and $\rho_{ijq}$
the recombination rate between the sets $a_{i}$ and $a_{j}$ resulting
in the set $a_{q}$, thus we have $\sum_{q=1}^{n}\rho_{ijq}=1$, then
the map $\varphi:A\otimes A'\rightarrow A$, $\varphi\left(a_{i}\otimes a_{k}^{-}\otimes a_{j}\right)=\frac{\theta}{2}\left(a_{i}a_{k}^{-}+a_{j}a_{k}^{-}\right)+\left(1-\theta\right)\sum_{q=1}^{n}\rho_{ijq}a_{q}\otimes a_{k}^{-}$
gives the distribution of egg types produced by a female of karyotype
$a_{i}\otimes a_{k}^{-}\otimes a_{j}$. If we note $\rho'_{ijq}$
and $\widetilde{\rho}_{ijq}$ the recombination rates between chromosome
sets $a_{i}$ and $a_{j}$ in spermatocytes of karyotype $a_{i}\otimes a_{k}^{-}\otimes\widetilde{a}_{j}$
that give sperm of $a'_{q}$ and $\widetilde{a}_{q}$ type, we have
$\sum_{q=1}^{n}\left(\rho'_{ijq}+\widetilde{\rho}_{ijq}\right)=1$.
Then $\varphi':A\otimes\widetilde{A}\rightarrow A'$, $\widetilde{\varphi}:A\otimes\widetilde{A}\rightarrow\widetilde{A}$
where $\varphi'\left(a_{i}\otimes a_{k}^{-}\otimes\widetilde{a}_{j}\right)=\sum_{q=1}^{n}\rho'_{ijq}a_{q}$
and $\widetilde{\varphi}\left(a_{i}\otimes a_{k}^{-}\otimes\widetilde{a}_{j}\right)=\sum_{q=1}^{m}\widetilde{\rho}_{ijq}\widetilde{a}_{q}$
respectively give the distributions of sperm carrying the $X$ and
$Y$ gonosome produced by a male of karyotype $a_{i}\otimes a_{k}^{-}\otimes\widetilde{a}_{j}$.\smallskip{}

\end{example}

\subsection{Construction from a baric algebra and two linear maps.}
\begin{prop}
\label{prop:Apond=0000262AL} Let $\left(A,\omega\right)$ be a baric
(not necessarily commutative) $K$-algebra and $\varphi,\widetilde{\varphi}:A\rightarrow A$
two linear maps such that $\omega\circ\left(\varphi+\widetilde{\varphi}\right)=\omega$.
Then the $K$-vector space $A\times A$ equipped with multiplication:
\[
\left(x,y\right)\left(x',y'\right)=\left(\varphi\left(xy'+x'y\right),\widetilde{\varphi}\left(xy'+x'y\right)\right)
\]
is a gonosomal algebra.\end{prop}
\begin{proof}
The algebra $A$ being weighted there is a basis $\left(a_{i}\right)_{1\leq i\leq n}$
in $A$ such that $a_{i}a_{j}=\sum_{k=1}^{n}\gamma_{ijk}a_{k}$ with
$\sum_{k=1}^{n}\gamma_{ijk}=1$ thus $\omega\left(a_{i}\right)=1$.
We put $e_{i}=\left(a_{i},0\right)$ and $\widetilde{e}_{i}=\left(0,a_{i}\right)$,
it is clear that $e_{i}e_{j}=\widetilde{e}_{i}\widetilde{e}_{j}=0$
for all $1\leq i,j\leq n$. Next with $\varphi\left(a_{k}\right)=\sum_{p=1}^{n}\alpha_{pk}a_{k}$
and $\widetilde{\varphi}\left(a_{k}\right)=\sum_{p=1}^{n}\widetilde{\alpha}_{pk}a_{k}$,
the assumption $\omega\circ\left(\varphi+\widetilde{\varphi}\right)=\omega$
is translated by $\sum_{p=1}^{n}\left(\alpha_{pk}+\widetilde{\alpha}_{pk}\right)=1$.
Then we have:
\begin{eqnarray*}
e_{i}\widetilde{e}_{j} & = & \left(\varphi\left(a_{i}a_{j}\right),\widetilde{\varphi}\left(a_{i}a_{j}\right)\right)\\
 & = & \left(\sum_{k=1}^{n}\gamma_{ijk}\varphi\left(a_{k}\right),\sum_{k=1}^{n}\gamma_{ijk}\widetilde{\varphi}\left(a_{k}\right)\right)\\
 & = & \sum_{k,p=1}^{n}\gamma_{ijk}\alpha_{pk}\left(a_{k},0\right)+\sum_{k,p=1}^{n}\gamma_{ijk}\widetilde{\alpha}_{pk}\left(0,a_{k}\right)
\end{eqnarray*}
with $\sum_{k,p=1}^{n}\gamma_{ijk}\left(\alpha_{pk}+\widetilde{\alpha}_{pk}\right)=1$,
thus the algebra $A\times A$ so defined is gonosomal.
\end{proof}
As illustrated in the following example, this construction is very
useful when sex determination is polygenic that is to say that sex
determination factors are distributed on several chromosomes.
\begin{example}
\emph{Sex determination in Musca domestica}.

In \emph{Musca domestica} $\left(2n=12\right)$, sex determination
follows the $XY$-system. On the $Y$ chromosome is a diallelic locus:
$M,+$. The allele $M$ determine male sex, allele $+$ is neutral.
The female sex is determined by an allele $F$ present on the autosome
IV, allele $F$ is recessive compared to $M$. Therefore male genotypes
are $\left(M+,FF\right)$ or $\left(MM,FF\right)$ and female $\left(++,FF\right)$.
However in some fly populations locus $M,+$ is also located on one
of the autosomes or even on the $X$ chromosome, in these populations
there is a mutation of the factor $F$ denoted by $F^{D}$ which is
dominant compared to $M$, in these cases the sex is determined by
18 genotypes including 10 for the sex female (see table below).

\medskip{}

\hspace{2cm}%
\begin{tabular}{ccccc}
\hline 
\multicolumn{2}{c}{Autosomes} & \multicolumn{3}{c}{Gonosomes}\tabularnewline
\hline 
IV & I-V & XX & XY & YY\tabularnewline
\hline 
$FF$ & $++$ & $\mbox{\Female}$ & $\mbox{\Male}$ & $\mbox{\Male}$\tabularnewline
$FF$ & $M+$, $\;MM$ & $\mbox{\Male}$ & $\mbox{\Male}$ & $\mbox{\Male}$\tabularnewline
$FF^{D}$ & $++$, $\;M+$, $\;MM$ & $\mbox{\Female}$ & $\mbox{\Female}$ & $\mbox{\Female}$\tabularnewline
\hline 
\end{tabular}

\medskip{}

Algebraically, given three spaces $A$, $B$ and $C$ with respective
bases $\left(a_{1},a_{2}\right)$, $\left(b_{1},b_{2}\right)$ and
$\left(c_{1},c_{2}\right)$ and with the algebra law: $xy=\frac{1}{2}x+\frac{1}{2}y$,
the space $A\otimes B\otimes C$ is equipped with the algebraic structure
$\left(a_{i}\otimes b_{j}\otimes c_{k}\right)\left(a_{p}\otimes b_{q}\otimes c_{r}\right)=\left(a_{i}a_{p}\right)\otimes\left(b_{j}b_{q}\right)\otimes\left(c_{k}c_{r}\right)$
and the weight function $\omega\left(a_{i}\otimes b_{j}\otimes c_{k}\right)=1$.
For $i,j,k\in\left\{ 1,2\right\} $ we note $e_{\left(i,j,k\right)}=a_{i}\otimes b_{j}\otimes c_{k}$,
and we put: 
\begin{eqnarray*}
\mbox{\Female} & = & \left\{ \begin{array}{c}
e_{\left(1,1,1\right)}\otimes e_{\left(1,1,1\right)},e_{\left(1,1,1\right)}\otimes e_{\left(2,1,1\right)},e_{\left(1,1,1\right)}\otimes e_{\left(2,2,1\right)},e_{\left(1,2,1\right)}\otimes e_{\left(2,2,1\right)},\\
e_{\left(1,1,1\right)}\otimes e_{\left(2,1,2\right)},e_{\left(1,1,1\right)}\otimes e_{\left(2,2,2\right)},e_{\left(1,2,1\right)}\otimes e_{\left(2,2,2\right)},e_{\left(1,1,2\right)}\otimes e_{\left(2,1,2\right)},\\
e_{\left(1,1,2\right)}\otimes e_{\left(2,2,2\right)},e_{\left(1,2,2\right)}\otimes e_{\left(2,2,2\right)}\hspace{4.8cm}
\end{array}\right\} \\
\mbox{\Male} & = & \left\{ \begin{array}{c}
e_{\left(1,1,1\right)}\otimes e_{\left(1,1,2\right)},e_{\left(1,1,2\right)}\otimes e_{\left(1,1,2\right)},e_{\left(1,1,1\right)}\otimes e_{\left(1,2,1\right)},e_{\left(1,2,1\right)}\otimes e_{\left(1,2,1\right)},\\
e_{\left(1,1,1\right)}\otimes e_{\left(1,2,2\right)},e_{\left(1,2,1\right)}\otimes e_{\left(1,2,2\right)},e_{\left(1,1,2\right)}\otimes e_{\left(1,2,2\right)},e_{\left(1,2,2\right)}\otimes e_{\left(1,2,2\right)}
\end{array}\right\} 
\end{eqnarray*}
then applying the proposition \ref{prop:Apond=0000262AL} with $\varphi,\widetilde{\varphi}$
defined on $A\otimes B\otimes C$ by $\varphi_{\left|\mbox{\Female}\right.}=id$,
$\varphi_{\left|\text{\mbox{\Male} }\right.}=0$ and $\widetilde{\varphi}_{\left|\mbox{\Female}\right.}=0$,
$\widetilde{\varphi}_{\left|\text{\mbox{\Male} }\right.}=id$, the
algebra $A\otimes B\otimes C$ is gonosomal. Using the coding $a_{1}\leftrightarrow F$,
$a_{2}\leftrightarrow F^{D}$, $b_{1}\leftrightarrow+$, $b_{2}\leftrightarrow M$,
$c_{1}\leftrightarrow X$, $c_{2}\leftrightarrow Y$, we obtain the
frequency distribution of crosses.
\end{example}
\smallskip{}

\begin{example}
\emph{Cytoplasmic heredity}.

Cytoplasmic heredity is the inheritance of genes that are not carried
by chromosomes in the nucleus cells, that is why it is called as extrachromosomal
or extranuclear heredity. These genes are located on the genome of
mitochondria (or chloroplasts in plants). Mitochondria are organelles
present in the cytoplasm, they play a key role in the production and
energy storage in eukaryotic cells. Mitochondria have their own DNA,
they divide continuously regardless of the division of the cell that
contains them, and they are distributed randomly between the daughter
cells during cell division. The mitochondrial genome mutates, this
translates into a cell by a genetically heterogeneous mitochondria
population and during oogenesis by a heterogeneous oocyte population:
some oocytes contain only one type of mitochondria (homoplasmy) others
contain several types (heteroplasmy). At the time of the sexual reproduction,
only the head of a sperm cell enters into the oocyte, it is possible
that some mitochondria of the flagellum are transferred at the same
time, however the number of mitochondria of paternal origin being
very low compared to those present in the oocyte, it is considered
that the mode of transmission of the mitochondrial genome is maternal.

Classify oocytes of a diploid sexual species in categories $c_{1},\ldots,c_{n}$
according to their mitochondrial populations. During fertilization,
the fusion of a $c_{i}$ type oocyte with a sperm gives a $c_{k}$
type egg with a frequency $\nu_{ki}$ (thus $\sum_{k=1}^{n}\nu_{ki}=1$),
this egg becomes female with a frequency $\sigma_{i}$ or male in
a proportion $1-\sigma_{i}$.

Algebraically, we apply the proposition \ref{prop:Apond=0000262AL}
on the algebra $A$ with basis $\left(c_{1},\ldots,c_{n}\right)$
defined by $c_{i}c_{j}=c_{i}$ for $1\leq i,j\leq n$, weighted by
$\omega\left(c_{i}\right)=1$ and the maps $\varphi\left(c_{i}\right)=\sigma_{i}\sum_{k=1}^{n}\nu_{ki}c_{k}$
and $\widetilde{\varphi}\left(c_{i}\right)=\left(1-\sigma_{i}\right)\sum_{k=1}^{n}\nu_{ki}c_{k}$.
\end{example}
\smallskip{}

\begin{example}
\emph{Sex determination by deuterotokous parthenogenesis}.

This type of parthenogenesis is known as cyclical because it alternates
a parthenogenetic phase and a sexual phase. In this case a diploid
female gives birth to diploid males or females.

In some species of aphids (Aphidoidea), the homogametic sex female
2A-2X are divided into two types: gynoparous females only lay eggs
that develop into females and androparous females only lay eggs that
develop into males. In gyniparous females a mitosis replaces meiosis
and produces 2A-2X genotype eggs that develop into females carrying
A-X type oocytes. In androparous female meiosis is abnormal: one of
the X chromosome is lost, the resulting spermatocytes 2A-X give males
with A-X type sperm (those of type A-0 are eliminated). The crosses
give only 2A-2X females.

Algebraically, we consider the algebra $V$ with basis $\left(a_{1},\ldots,a_{n}\right)$,
multiplication $a_{i}a_{j}=\frac{1}{2}\left(a_{i}+a_{j}\right)$ and
weight function $\omega\left(a_{i}\right)=1$. Let $A=D\left(V\right)$
be the (commutative or non commutative) duplicate of $V$, we apply
the proposition \ref{prop:Apond=0000262AL} with $\varphi\left(x\right)=\theta x$,
$\widetilde{\varphi}\left(x\right)=\left(1-\theta\right)x$. Let $e_{ij}=\left(a_{i}\otimes a_{j},0\right)$
and $\widetilde{e}_{ij}=\left(0,a_{i}\otimes a_{j}\right)$ we obtain:
\[
e_{ij}\widetilde{e}_{pq}=\tfrac{\theta}{8}\left(\bigl(e_{i}+e_{j}\bigr)\otimes\bigl(e_{p}+e_{q}\bigr),0\right)+\tfrac{1-\theta}{8}\left(0,\bigl(e_{i}+e_{j}\bigr)\otimes\bigl(e_{p}+e_{q}\bigr)\right).
\]

\end{example}
\smallskip{}

\subsection{Construction from a baric algebra and a gonosomal algebra.}
\begin{prop}
\label{prop:AG=000026PondGono} Let $\left(A,\omega\right)$ be a
finite-dimensional (not necessarily commutative) baric $K$-algebra
and $G$ a gonosomal $K$-algebra. The space $A\otimes G$ equipped
with the multiplication
\[
\left(x\otimes y\right)\left(x'\otimes y'\right)=\frac{1}{2}\left(xx'+x'x\right)\otimes yy',
\]
is a gonosomal algebra.\end{prop}
\begin{proof}
The algebra $A$ admits a basis $\left(a_{1},\ldots,a_{N}\right)$
such that $a_{p}a_{q}=\sum_{r=1}^{N}\lambda_{pqr}a_{r}$ with $\sum_{r=1}^{N}\lambda_{pqr}=1$,
and $G$ has a gonosomal basis $\left\{ e_{i};1\leq i\leq n\right\} \cup\left\{ \widetilde{e}_{j};1\leq j\leq m\right\} $.
We put $a_{p,i}=a_{p}\otimes e_{i}$ and $\widetilde{a}_{p,j}=a_{p}\otimes\widetilde{e}_{j}$
for $1\leq p\leq N$, $1\leq i\leq n$ and $1\leq j\leq m$. We have
$a_{p,i}a_{q,j}=0$, $\widetilde{a}_{p,i}\widetilde{a}_{q,j}=0$ and
\[
a_{p,i}\widetilde{a}_{q,j}=\frac{1}{2}\left(\sum_{r,k}\left(\lambda_{pqr}+\lambda_{qpr}\right)\gamma_{ijk}\right)a_{r,k}+\frac{1}{2}\left(\sum_{k,r}\left(\lambda_{pqr}+\lambda_{qpr}\right)\widetilde{\gamma}_{ijk}\right)\widetilde{a}_{r,k}
\]
with
\begin{eqnarray*}
\tfrac{1}{2}\left(\sum_{r,k}\left(\lambda_{pqr}+\lambda_{qpr}\right)\gamma_{ijk}\right)+\tfrac{1}{2}\left(\sum_{k,r}\left(\lambda_{pqr}+\lambda_{qpr}\right)\widetilde{\gamma}_{ijk}\right) & = & 1.
\end{eqnarray*}
This proves that $A\otimes G$ is a gonosomal algebra.\smallskip{}
\end{proof}
\begin{example}
\emph{Crosses with distribution of autosomal types according to sexes}.
\end{example}
In a bisexual population we consider autosomal genetic types of autosomal
$a_{1},\ldots,a_{n}$ and we assume that over the generations the
proportion of females in the population has a constant constant value
$0\leq\sigma\leq1$. This frequency $\sigma$ is connected to the
population sex-ratio $\rho$ (proportion of males to females per generation),
indeed we have $\rho=\frac{1-\sigma}{\sigma}$ thus $\sigma=\frac{1}{1+\rho}$.

Let $\gamma_{ijk}$ the frequency of type $a_{k}$ in the progeny
of two individuals of types $a_{i}$ and $a_{j}$. Let $A$ be the
algebra with basis $\left(a_{1},\ldots,a_{n}\right)$, the commutative
law $a_{i}a_{j}=\sum_{k=1}^{n}\gamma_{ijk}a_{k}$, $\sum_{k=1}^{n}\gamma_{ijk}=1$
and $G$ the gonosomal algebra with basis $\left(f,m\right)$ defined
by $f^{2}=m^{2}=0$ and $fm=mf=\sigma f+\left(1-\sigma\right)m$.
If we apply the proposition \ref{prop:AG=000026PondGono} to $A\otimes G$
we find $\left(a_{i}\otimes f\right)\left(a_{j}\otimes m\right)=\sigma\sum_{k=1}^{n}\gamma_{ijk}a_{k}\otimes f+\left(1-\sigma\right)\sum_{k=1}^{n}\gamma_{ijk}a_{k}\otimes m$,
therefore after the cross of types $a_{i}$ and $a_{j}$, the frequency
of type $a_{k}$ in the female (resp. male) offspring is $\sigma\gamma_{ijk}$
(resp. $\left(1-\sigma\right)\gamma_{ijk}$).

\smallskip{}

\begin{example}
\emph{Transmission law of a couple of autosomal and gonosomal genes}.
\end{example}
We consider in a bisexual population the autosomal types $a_{1},\ldots,a_{N}$;
the female $e_{1},\ldots,e_{n}$ and male $\widetilde{e}_{1},\ldots,\widetilde{e}_{m}$
gonosomal types. Let $A$ be algebra with basis $\left(a_{i}\right)_{1\leq i\leq N}$
with commutative law $a_{i}a_{j}=\sum_{k=1}^{N}\lambda_{ijk}a_{k}$,
$\sum_{k=1}^{N}\lambda_{ijk}=1$ and $G$ the gonosomal algebra with
basis $\left(e_{1},\ldots,e_{n},\widetilde{e}_{1},\ldots,\widetilde{e}_{m}\right)$
defined by $e_{i}\widetilde{e}_{j}=\sum_{k=1}^{n}\gamma_{ijk}e_{i}+\sum_{k=1}^{m}\widetilde{\gamma}_{ijk}\widetilde{e}_{k}$.
If we apply the proposition \ref{prop:AG=000026PondGono} to $A\otimes G$
we find:
\[
\left(a_{i}\otimes e_{p}\right)\left(a_{j}\otimes\widetilde{e}_{q}\right)=\sum_{k=1}^{N}\left(\sum_{r=1}^{n}\lambda_{ijk}\gamma_{pqr}a_{k}\otimes e_{r}+\sum_{r=1}^{m}\lambda_{ijk}\widetilde{\gamma}_{pqr}a_{k}\otimes\widetilde{e}_{r}\right),
\]
in other words, after crossing a $\left(a_{i},e_{p}\right)$ female
type with a $\left(a_{j},\widetilde{e}_{q}\right)$ male the frequency
of type $\left(a_{k},e_{r}\right)$ in females is $\lambda_{ijk}\gamma_{pqr}$
and type $\left(a_{k},\widetilde{e}_{r}\right)$ in males is $\lambda_{ijk}\widetilde{\gamma}_{pqr}$.

\section{On baricity and dibaricity of gonosomal algebras}

The notion of baric algebra (see definition in section 4) was introduced
by Etherington \cite{Ether-39}, it plays a fundamental role in the
study of genetic algebras. But as zygotic algebras for sex-linked
inheritance admit bases whose elements are nilpotent this means they
are not baric. In order to remedy this, Holgate \cite{Holg-70} introduced
the notion of dibaric algebra. 
\begin{defn}
A $K$-algebra $A$ is dibaric if there is a surjective morphism $\chi:A\rightarrow\mathcal{S}$
where $\mathcal{S}$ is the $K$-algebra with basis $\left\{ f,m\right\} $
such that $f^{2}=m^{2}=0$ and $fm=mf=\frac{1}{2}\left(f+m\right)$.\end{defn}
\begin{rem}
A $K$-algebra with basis $\left\{ a,b\right\} $ such that $a^{2}=b^{2}=0$,
$ab=ba=\theta a+\left(1-\theta\right)b$ where $\theta\in K$, $\theta\neq0,1$,
is isomorphic to the dibaric algebra $\mathcal{S}$ with basis $\left\{ f,m\right\} $
by $a\mapsto2\left(1-\theta\right)f$ and $b\mapsto2\theta m$.
\end{rem}
If $A$ is a dibaric algebra then $A^{2}$ is baric. There are other
results concerning dibaric algebras in \cite{CF-00} and \cite{LOR-11}.
\begin{rem}
In \cite{LR-10} the authors showed that the evolution algebra of
the bisexual population are dibaric. The example \ref{exa:hemophilia}
shows that gonosomal algebras are not in general dibaric. Indeed,
let us assume that the algebra $A$ defined in example \ref{exa:hemophilia}
is dibaric, let $\chi:A\rightarrow\mathcal{S}$ a dibaric function
such that $\chi\left(e_{i}\right)=\lambda_{i}f+\mu_{i}m$ and $\chi\left(\widetilde{e}_{i}\right)=\widetilde{\lambda}_{i}f+\widetilde{\mu}_{i}m$
for all $i=1,2$ . 

From $\chi\left(e_{i}\right)\chi\left(e_{j}\right)=\chi\left(e_{i}e_{j}\right)=0$
and $\chi\left(\widetilde{e}_{p}\right)\chi\left(\widetilde{e}_{q}\right)=\chi\left(\widetilde{e}_{p}\widetilde{e}_{q}\right)=0$
it follows 6 identities: 
\begin{eqnarray}
\lambda_{i}\mu_{j}+\mu_{i}\lambda_{j} & = & 0,\quad1\leq i\leq j\leq2,\label{eq:E1-1}\\
\widetilde{\lambda}_{p}\widetilde{\mu}_{q}+\widetilde{\mu}_{p}\widetilde{\lambda}_{q} & = & 0,\quad1\leq p\leq q\leq2.\label{eq:E2-1}
\end{eqnarray}
 From $\eta\left(e_{i}\widetilde{e}_{p}\right)=\eta\left(e_{i}\right)\eta\left(\widetilde{e}_{p}\right)$
it follows 8 identities:

\begin{flushleft}
\begin{tabular}{ll}
(1) : $\lambda_{1}+\widetilde{\lambda}_{1}=\lambda_{1}\widetilde{\mu}_{1}+\mu_{1}\widetilde{\lambda}_{1}$, & (5) : $\mu_{1}+\widetilde{\mu}_{1}=\lambda_{1}\widetilde{\mu}_{1}+\mu_{1}\widetilde{\lambda}_{1}$,\tabularnewline
(2) : $\lambda_{2}+\widetilde{\lambda}_{1}=\lambda_{1}\widetilde{\mu}_{2}+\mu_{1}\widetilde{\lambda}_{2}$, & (6) : $\mu_{2}+\widetilde{\mu}_{1}=\lambda_{1}\widetilde{\mu}_{2}+\mu_{1}\widetilde{\lambda}_{2}$,\tabularnewline
(3) : $\lambda_{1}+\lambda_{2}+\widetilde{\lambda}_{1}+\widetilde{\lambda}_{2}=2\left(\lambda_{2}\widetilde{\mu}_{1}+\mu_{2}\widetilde{\lambda}_{1}\right)$, & (7) : $\mu_{1}+\mu_{2}+\widetilde{\mu}_{1}+\widetilde{\mu}_{2}=2\left(\lambda_{2}\widetilde{\mu}_{1}+\mu_{2}\widetilde{\lambda}_{1}\right)$,\tabularnewline
(4) : $\lambda_{2}+\widetilde{\lambda}_{1}+\widetilde{\lambda}_{2}=\frac{3}{2}\left(\lambda_{2}\widetilde{\mu}_{2}+\mu_{2}\widetilde{\lambda}_{2}\right)$, & (8) : $\mu_{2}+\widetilde{\mu}_{1}+\widetilde{\mu}_{2}=\frac{3}{2}\left(\lambda_{2}\widetilde{\mu}_{2}+\mu_{2}\widetilde{\lambda}_{2}\right)$.\tabularnewline
\end{tabular} 
\par\end{flushleft}

With (\ref{eq:E1-1}) we have $\lambda_{1}\mu_{1}=0$. If $\lambda_{1}\neq0$
then $\mu_{1}=0$ what implies in (\ref{eq:E1-1}) $\lambda_{1}\mu_{2}=0$
thus $\mu_{2}=0$, this combined with relations (5) and (6) gives
$\widetilde{\mu}_{1}=\widetilde{\mu}_{2}$, but from (3) and (4) it
comes $\lambda_{1}=2\lambda_{2}\widetilde{\mu}_{1}-\frac{3}{2}\lambda_{2}\widetilde{\mu}_{2}=\frac{1}{2}\lambda_{2}\widetilde{\mu}_{1}$
 and between the relations (7) and (8) we deduct $\lambda_{2}\widetilde{\mu}_{1}=0$,
hence a contradiction. 

Thus we have $\lambda_{1}=0$. If $\lambda_{2}\neq0$, by (\ref{eq:E1-1})
it comes $\mu_{1}=\mu_{2}=0$ which results in (1) that $\widetilde{\lambda}_{1}=0$
hence by (2) we have $\lambda_{2}=0$, contradiction. 

Thus we have $\lambda_{1}=\lambda_{2}=0$, this implies by difference
between (1) and (2) that $\mu_{1}(\widetilde{\lambda}_{1}-\widetilde{\lambda}_{2})=0$,
and if we have $\mu_{1}\neq0$ then $\widetilde{\lambda}_{1}=\widetilde{\lambda}_{2}$,
by difference between (3) and (4) it comes $\mu_{2}\widetilde{\lambda}_{1}=0$,
and by difference between (7) and (8) we get to $\mu_{1}=0$, contradiction.
Therefore we have $\mu_{1}=0$, this results in (1) that $\widetilde{\lambda}_{1}=0$,
we deduct from (3) that $\widetilde{\lambda}_{2}=0$, by (5) that
$\widetilde{\mu}_{1}=0$, in (6) we have $\mu_{2}=0$ and in (7) we
get $\widetilde{\mu}_{2}=0$, finally we have $\chi=0$, contradiction.

In conclusion we can not define a dibaric function on $A$.\end{rem}
\begin{prop}
Let $A$ be a gonosomal $K$-algebra with a gonosomal basis $\left(e_{i}\right)_{1\leq i\leq n}\cup\bigl(\widetilde{e}_{p}\bigr)_{1\leq p\leq m}$
such that $e_{i}\widetilde{e}_{j}=\sum_{k=1}^{n}\gamma_{ijk}e_{k}+\sum_{p=1}^{m}\widetilde{\gamma}_{ijp}\widetilde{e}_{p}$,
if it exists $\sigma\in K$, $\sigma\neq0,1$ verifying $\sum_{k=1}^{n}\gamma_{ipk}=\sigma$
for all $1\leq i\leq n$ and $1\leq p\leq m$ then $A$ is dibaric.\end{prop}
\begin{proof}
Indeed, the algebra $A$ is diweighted by $\chi\left(e_{i}\right)=\frac{1}{2\sigma}f$
and $\chi\left(\widetilde{e}_{p}\right)=\frac{1}{2\left(1-\sigma\right)}m$.\end{proof}
\begin{prop}
Let $A$ be a gonosomal algebra with gonosomal basis $\left(e_{i}\right)_{1\leq i\leq n}\cup\bigl(\widetilde{e}_{p}\bigr)_{1\leq p\leq m}$
and multiplication $e_{i}\widetilde{e}_{j}=\sum_{k=1}^{n}\gamma_{ijk}e_{k}+\sum_{p=1}^{m}\widetilde{\gamma}_{ijp}\widetilde{e}_{p}$.
The algebra $A$ is dibaric if and only if the system of $2nm$ quadratic
equations in unknowns $x_{1},\ldots,x_{n}$, $y_{1},\ldots,y_{m}$
:
\[
\begin{cases}
x_{i}y_{j}-2\sum_{p=1}^{n}\gamma_{ijp}x_{p} & =0\medskip\\
x_{i}y_{j}-2\sum_{q=1}^{m}\widetilde{\gamma}_{ijq}y_{p} & =0,\quad\left(1\leq i\leq n,1\leq j\leq m\right)
\end{cases}
\]
admits a non-zero solution.\end{prop}
\begin{proof}
Suppose that $A$ is diweighted by $\chi$ with $\chi\left(e_{i}\right)=\alpha_{i}f+\beta_{i}m$
and $\chi\left(\widetilde{e}_{i}\right)=\widetilde{\alpha}_{i}f+\widetilde{\beta}_{i}m$.
From $\chi\left(e_{i}e_{j}\right)=\chi\left(e_{i}\right)\chi\left(e_{j}\right)=0$
and $\chi\left(\widetilde{e}_{p}\widetilde{e}_{q}\right)=\chi\left(\widetilde{e}_{p}\right)\chi\left(\widetilde{e}_{q}\right)=0$
it comes: 
\begin{eqnarray}
\alpha_{i}\beta_{j}+\beta_{i}\alpha_{j} & = & 0,\quad\left(1\leq i,j\leq n\right)\label{eq:eiej}\\
\widetilde{\alpha}_{p}\widetilde{\beta}_{q}+\widetilde{\beta}_{p}\widetilde{\alpha}_{q} & = & 0,\quad\left(1\leq p,q\leq m\right).\label{eq:etietj}
\end{eqnarray}
From $\chi\left(e_{i}\widetilde{e}_{j}\right)=\chi\left(e_{i}\right)\chi\left(\widetilde{e}_{j}\right)$
it results:
\begin{equation}
\begin{cases}
\sum_{p=1}^{n}\gamma_{ijp}\alpha_{p}+\sum_{q=1}^{m}\widetilde{\gamma}_{ijq}\widetilde{\alpha}_{p} & =\tfrac{1}{2}\left(\alpha_{i}\widetilde{\beta}_{j}+\beta_{i}\widetilde{\alpha}_{j}\right),\medskip\\
\sum_{p=1}^{n}\gamma_{ijp}\beta_{p}+\sum_{q=1}^{m}\widetilde{\gamma}_{ijq}\widetilde{\beta}_{p} & =\tfrac{1}{2}\left(\alpha_{i}\widetilde{\beta}_{j}+\beta_{i}\widetilde{\alpha}_{j}\right),\quad\left(1\leq i\leq n,1\leq j\leq m\right).
\end{cases}\label{eq:eietj}
\end{equation}
 There is $1\leq i_{0}\leq n$ or $1\leq j_{0}\leq m$ such that $\alpha_{i_{0}}\neq0$
or $\widetilde{\alpha}_{j_{0}}\neq0$ because otherwise it would $\chi\left(A\right)=K\bigl\langle m\bigr\rangle$
and $\chi$ would not surjective. Suppose for example that $\alpha_{i_{0}}\neq0$
for an integer $1\leq i_{0}\leq n$, then from the equations (\ref{eq:eiej})
it comes $\alpha_{i_{0}}\beta_{i_{0}}=0$ from where $\beta_{i_{0}}=0$,
it follows from (\ref{eq:eiej}) that for all $j\neq i_{0}$ we have
$\alpha_{i_{0}}\beta_{j}=0$ thus $\beta_{j}=0$ for all $1\leq j\leq m$,
we deduce that there is $1\leq j_{0}\leq m$ such that $\widetilde{\beta}_{j_{0}}\neq0$
otherwise we would have $\chi\left(A\right)=K\bigl\langle f\bigr\rangle$
and $\chi$ would not surjective and by the same way as above this
leads to $\widetilde{\alpha}_{j}=0$ for all $1\leq j\leq m$. It
follows that the system (\ref{eq:eietj}) becomes:
\begin{equation}
\begin{cases}
\sum_{p=1}^{n}\gamma_{ijp}\alpha_{p} & =\tfrac{1}{2}\alpha_{i}\widetilde{\beta}_{j},\medskip\\
\sum_{q=1}^{m}\widetilde{\gamma}_{ijq}\widetilde{\beta}_{p} & =\tfrac{1}{2}\alpha_{i}\widetilde{\beta}_{j},\quad\left(1\leq i\leq n,1\leq j\leq m\right)
\end{cases}\label{eq:eietj-1}
\end{equation}
Conversely, if the system (\ref{eq:eietj-1}) admits a non-zero solution,
then we verifie without difficulty the map $\chi$ defined on $A$
by $\chi\left(e_{i}\right)=\alpha_{i}f$ and $\chi\left(\widetilde{e}_{i}\right)=\widetilde{\beta}_{i}m$
is a dibaric function.
\end{proof}

\section{Idempotents in gonosomal algebras}

In order to study the idempotents of a gonosomal algebra we introduce
the following definition:
\begin{defn}
Let $A$ be a gonosomal $K$-algebra with gonosomal basis $\left(e_{i}\right)_{1\leq i\leq n}\cup\left(\widetilde{e}_{j}\right)_{1\leq j\leq m}$.
For $x=\sum_{i=1}^{n}\alpha_{i}e_{i}+\sum_{i=1}^{m}\widetilde{\alpha}_{i}\widetilde{e}_{i}$
in $A$, the scalar $\mu\left(x\right)=\sum_{i=1}^{n}\alpha_{i}+\sum_{i=1}^{m}\widetilde{\alpha}_{i}$
is called mass of $x$.\end{defn}
\begin{prop}
For a gonosomal $K$-algebra admits an idempotent of mass $\mu\neq0$
it is necessary that the field $K$ contains the roots of $X^{2}-\left(\mu-1\right)^{2}+1$.\end{prop}
\begin{proof}
Let $A$ be a gonosomal $K$-algebra with gonosomal basis $\left(e_{i}\right)_{1\leq i\leq n}\cup\left(\widetilde{e}_{j}\right)_{1\leq j\leq m}$.
Let us suppose that there exists in $A$ an idempotent $e=\sum_{i=1}^{n}\alpha_{i}e_{i}+\sum_{i=1}^{m}\widetilde{\alpha}_{i}\widetilde{e}_{i}$
of mass $\mu\left(e\right)=\mu$. We have $e^{2}=e$ if and only if
we have $2\sum_{i,j}\alpha_{i}\widetilde{\alpha}_{j}e_{i}\widetilde{e}_{j}=e$,
which results in the system of quadratic equations with $n+m$ unknowns
$\alpha_{1},\ldots,\alpha_{n}$ and $\widetilde{\alpha}_{1},\ldots,\widetilde{\alpha}_{m}$
:
\begin{equation}
\begin{cases}
2\sum_{i,j}\gamma_{ijp}\alpha_{i}\widetilde{\alpha}_{j}=\alpha_{p}, & 1\leq p\leq n,\\
2\sum_{i,j}\widetilde{\gamma}_{ijq}\alpha_{i}\widetilde{\alpha}_{j}=\widetilde{\alpha}_{q}, & 1\leq q\leq m.
\end{cases}\label{eq:syst-idempotent}
\end{equation}
If we make the sum of all equations of this system we obtain: 
\[
2\left(\sum_{i=1}^{n}\alpha_{i}\right)\left(\sum_{j=1}^{m}\widetilde{\alpha}_{j}\right)=\sum_{p=1}^{n}\alpha_{p}+\sum_{q=1}^{m}\widetilde{\alpha}_{q}=\mu,
\]
in other words $\sum_{p=1}^{n}\alpha_{p}$ and $\sum_{q=1}^{m}\widetilde{\alpha}_{q}$
are roots of $2X^{2}-2\mu X+\mu$ whose discriminant is $\mu^{2}-2\mu=\left(\mu-1\right)^{2}-1$.
\end{proof}
In the particular case of a population with a dominant male where
a single male involved in reproduction, we can specify the set of
idempotents. This case was studied in \cite{LLR-14} and \cite{LR-arx}
for evolution algebras of the bisexual population, for gonosomal algebras
we have:
\begin{prop}
Let $A$ be a gonosomal algebra with gonosomal basis $\left(e_{i}\right)_{1\leq i\leq n}\cup\left(\widetilde{e}\right)$
where $e_{i}\widetilde{e}=\sum_{p=1}^{n}\gamma_{ip}e_{p}+\widetilde{\gamma}_{i}\widetilde{e}$.
Noting $\Gamma=\left(\gamma_{ip}\right)_{1\leq i,p\leq n}\in\mathbf{M}_{n}\left(K\right)$
and $\widetilde{\Gamma}=\left(\widetilde{\gamma}_{i}\right)_{1\leq i\leq n}\in\mathbf{M}_{1,n}\left(K\right)$,
the set of nonzero idempotents of $A$ is $\left\{ \left(\alpha,\widetilde{\alpha}\right)\in K^{n}\times K\right\} ,$
where 

\qquad{}$\widetilde{\alpha}\neq0$ is such that $\left(2\widetilde{\alpha}\right)^{-1}\in\mbox{Spec}\left(\Gamma\right)$, 

\qquad{}$\alpha\in K^{n}$ verify $\left(\Gamma-\left(2\widetilde{\alpha}\right)^{-1}I_{n}\right)\alpha^{T}=0$
and $\widetilde{\Gamma}\alpha^{T}=\tfrac{1}{2}$. \end{prop}
\begin{proof}
We have $n\geq2$ and $m=1$, let $x=\sum_{i=1}^{n}\alpha_{i}e_{i}+\widetilde{\alpha}\widetilde{e}$
be an idempotent of $A$, with this the system (\ref{eq:syst-idempotent})
is written: 
\[
\begin{cases}
2\widetilde{\alpha}\sum_{i=1}^{m}\gamma_{ip}\alpha_{i}=\alpha_{p}, & 1\leq p\leq n,\\
2\widetilde{\alpha}\sum_{i=1}^{m}\widetilde{\gamma}_{i}\alpha_{i}=\widetilde{\alpha}.
\end{cases}\mbox{\qquad\ensuremath{\left(*\right)}}
\]

If $\widetilde{\alpha}=0$ we have $\alpha_{1}=\ldots=\alpha_{n}=0$.

If $\widetilde{\alpha}\neq0$, puting $\Gamma=\left(\gamma_{ip}\right)_{1\leq i,p\leq n}\in\mathbf{M}_{n}\left(K\right)$,
$\widetilde{\Gamma}=\left(\widetilde{\gamma}_{i}\right)_{1\leq i\leq n}\in\mathbf{M}_{1,n}\left(K\right)$
and $\alpha=\left(\alpha_{1},\ldots,\alpha_{n}\right)$, the above
linear system becomes:
\begin{eqnarray*}
\left(\Gamma-\left(2\widetilde{\alpha}\right)^{-1}I_{n}\right)\alpha^{T}=0, &  & \widetilde{\Gamma}\alpha^{T}=\tfrac{1}{2}\quad\left(*\right).
\end{eqnarray*}
We have $\mbox{det}\left(\Gamma-\left(2\widetilde{\alpha}\right)^{-1}I_{n}\right)=0$,
otherwise the system $\left(*\right)$ has no solution, therefore
$\left(2\widetilde{\alpha}\right)^{-1}\in Spec\left(\Gamma\right)$
and $\alpha^{T}\in\ker\left(\Gamma-\left(2\widetilde{\alpha}\right)^{-1}I_{n}\right)$. 
\end{proof}
\bigskip{}

A presentation of a part of this work was given to \textquotedblleft VI
International Conference on non associative algebra and its applications\textquotedblright{}
(Zaragoza, november 2011). \bigskip{}


\begin{thebibliography}{10}
\bibitem{CF-00} M. A. Couto, J. C. Guti\'{e}rrez Fern\'{a}ndez.
Dibaric algebras. \emph{Proyecciones} \textbf{19} (3) : 249--269 (2000).

\bibitem{DA-LS-PV-89}A. I. Durand-Alegria, J. Lopez-Sanchez, A. Perez
de Vargas. \emph{Zygotic algebra for two-linked loci with sexually
different recombination and mutation rates}. Lin. Alg. Appl. \textbf{121}
: 385--399 (1989). 

\bibitem{Ether-39}I.M.H. Etherington. \emph{Genetic algebras}. Proc.
Roy. Soc. Edinburgh. \textbf{59} : 242--258 (1939) . 

\bibitem{Ether-41} I.M.H. Etherington. \textsl{Non associative algebra}\emph{
and the symbolism of genetics}. Proc. Roy. Soc. Edinburgh. \textbf{61}
: 24--42 (1941).

\bibitem{Gonsh-60}H. Gonshor. \emph{Special train algebra arising
in genetics}. Proc. Edinburgh Math. Soc. (1) \textbf{12} : 41--53
(1960).

\bibitem{Gonsh-65}H. Gonshor. \emph{Special train algebra arising
in genetics II}. Proc. Edinburgh Math. Soc. \textbf{14} (4) : 333--338
(1965).

\bibitem{Gonsh-73} H. Gonshor. \emph{Contributions to genetic algebra
II}. Proc. Edinburgh Math. Soc. 18 (4) : 273--279 (1973).

\bibitem{Heuch-75}I. Heuch. \emph{Partial and complete sex linkage
in infinite populations}. J. Math. Biol. \textbf{1} : 331--343 (1975).

\bibitem{Holg-70}P. Holgate. \textsl{Genetic algebra associated with
sex linkage}. Proc. Edinburgh Math. Soc. \textbf{17} : 113--120 (1970).

\bibitem{LLR-14} A. Labra, M. Ladra, U. A. Rozikov. \emph{An evolution
algebra in population genetics}. Lin. Alg. Appl. \textbf{457} : 348\textendash 362,
(2014).

\bibitem{LR-10}M. Ladra and U. A. Rozikov. \emph{Evolution algebra
of a bisexual population}. J. Algebra \textbf{378} : 153-172 (2013). 

\bibitem{LR-arx} M. Ladra, U. A. Rozikov. Evolution algebra of a
\textquotedblleft chicken\textquotedblright{} population, arXiv:1307.4916.

\bibitem{LOR-11}M. Ladra, B. A. Omirov and U. A. Rozikov. \emph{On
dibaric and evolution algebras}. \texttt{arXiv:1104.2578v1} (13 Apr
2011).

\bibitem{WB-74}A. W\"orz-Busekros. \emph{The zygotic algebra for
sex linkage}. J. Math. Biol. \textbf{1 }: 37--46 (1974).

\bibitem{WB-75}A. W\"orz-Busekros. \emph{The zygotic algebra for
sex linkage} \emph{II}. J. Math. Biol. \textbf{2 }: 359--371 (1975).

\bibitem{WB-80}A. W\"orz-Busekros. ``Algebras in Genetics''. Lecture
Notes in Biomathematics, \textbf{36}. Springer-Verlag, New York, 1980.\end{thebibliography}
\end{document}